\documentclass[12pt]{article}
\usepackage{enumitem}

\usepackage[margin=1in]{geometry}

\usepackage{natbib}

\usepackage[T1]{fontenc}

\usepackage{placeins}

\usepackage{setspace}
\usepackage{amsmath, amssymb, amsthm, mathtools}

\theoremstyle{plain}
\newtheorem{proposition}{Proposition}
\newtheorem{lemma}{Lemma}

\theoremstyle{definition}
\newtheorem{definition}{Definition}
\newtheorem{assumption}{Assumption}

\theoremstyle{remark}
\newtheorem{remark}{Remark}

\usepackage{graphicx}
\usepackage{float}
\usepackage{booktabs}
\usepackage{caption}
\usepackage{subcaption}
\usepackage[colorlinks=true, citecolor=blue, linkcolor=blue, urlcolor=blue]{hyperref}

\newcommand{\dP}{\partial P}
\newcommand{\dQD}{\partial Q^D}
\newcommand{\dS}{\partial S}
\newcommand{\Slow}{S_{\text{low}}}
\newcommand{\Shigh}{S_{\text{high}}}

\newcommand{\Sunstable}{S_{\text{unstable}}}
\newcommand{\SFB}{S^{FB}}
\newcommand{\Ssmall}{S_{\text{small}}}

\begin{document}

\title{\textbf{Coordination Failures and Stackelberg Leadership in Housing Development with Network Effects}}

\author{
Northwestern University \\
Mathematical Methods in the Social Sciences \\
Vaibhav Rangan\thanks{I would like to thank my advisor, Professor Ben Golub, for his invaluable and tireless guidance and support on my project, as well as for his generous provision of Refine.ink credits. Email: \texttt{vaibhavrangan2027@u.northwestern.edu}.}
}
\date{May 2026}

\maketitle

\begin{abstract}
\noindent I study coordination failures in housing development markets with network effects, where the value of building depends on aggregate supply. When network effects are sufficiently strong and convex, multiple equilibria arise: a low-supply coordination failure and a high-supply outcome. Without a coordination mechanism, equilibrium is indeterminate. I introduce a large developer who moves first in a Stackelberg game, committing to housing supply before atomistic developers make entry decisions. The main result is that the large developer always commits at least to the high-supply equilibrium, eliminating the coordination failure by pushing past the unstable threshold that separates the low and high outcomes. The result is unconditional; it holds for general demand functions and cost distributions, and does not depend on which stable continuation equilibrium materializes. The leader's commitment inverts standard monopoly intuition: first-mover commitment can improve welfare by resolving a coordination problem that atomistic markets cannot solve on their own. I also characterize when the developer builds beyond the high equilibrium into a monopoly region, and show that the market underprovides housing relative to the social optimum.



\end{abstract}

\newpage
\tableofcontents
\newpage

\section{Introduction}
\label{sec:intro}
Housing development is fundamentally a coordination game. Developers commit large amounts of up-front capital before knowing whether residents will follow, and prospective residents value density. When more people move into an area, businesses, transit, public services, and other amenities follow, raising willingness to pay for all units \citep{rossi_hansberg_sarte_owens_2010, guerrieri_hartley_hurst_2013}. Each developer's return therefore depends on how many other developers build, but no individual developer controls or observes others' plans before committing. Competitor entry operates as a positive payoff complementarity rather than a signal of latent demand.

When network effects are sufficiently strong and convex, this coordination problem generates multiple equilibria. A high-supply equilibrium exists in which density rises, amenities follow, and the market sustains itself. A low-supply equilibrium also exists, in which insufficient density makes entry unprofitable and pessimism persists. The same fundamentals are consistent with both outcomes, and the market provides no mechanism to select between them. This indeterminacy has real consequences. \citet{freemark_2020} documents that transit-oriented upzoning in Chicago produced no increase in construction over five years despite rising property values, consistent with regulatory relaxation capitalizing into land prices without triggering coordinated development. By contrast, Auckland's 2016 upzoning of three-quarters of its residential land was followed by a sustained construction boom, adding an estimated 20,000 dwellings in five years \citep{greenaway_mcgrevy_phillips_2023}. The divergence is difficult to explain through regulation alone; a coordination-based account, in which Auckland crossed a threshold that Chicago did not, offers a natural explanation.

This paper formalizes the coordination problem as a two-stage game of complete information. A neighborhood contains a fixed stock of buildable sites with publicly observable construction costs. A single large developer with the scale to assemble sites in advance moves first, choosing which sites to develop, followed by a continuum of atomistic developers who build the remaining sites. Under sufficient convexity of network effects, the Stage 2 game admits multiple stable equilibria. I frame the large developer as a Stackelberg leader and ask whether their commitment can eliminate the coordination failure. The main result is that the large developer strictly prefers committing at least to the high-supply level, which requires pushing past the unstable threshold separating the low and high equilibria. The result is unconditional: the developer's optimal commitment dominates every alternative under every possible Stage 2 realization, requiring no assumption about which stable equilibrium materializes.

Resolving the coordination failure with a large mover yields tangible welfare benefits. I show that any atomistic market equilibrium underprovisions housing relative to the social optimum, because developers cannot charge existing residents for the network benefits their construction generates. The Stackelberg leader partially closes this gap by eliminating the low equilibrium, strictly improving welfare. However, even the high equilibrium falls short of the planner's optimum since the network externality remains uninternalized. I further characterize conditions under which the large developer builds beyond the high equilibrium into a monopoly region, and show that this occurs precisely when network effects dominate the standard demand-depressing effect of additional supply.

The remainder of the paper is organized as follows. Section~\ref{sec:literature} reviews related literature on coordination games, housing network effects, and Stackelberg models. Section~\ref{sec:model} presents the model. Section~\ref{sec:atomistic} characterizes the atomistic equilibrium and its multiplicity. Section~\ref{sec:stackelberg} defines the Stackelberg equilibrium, states the main results, and proves them. Section~\ref{sec:welfare} develops the welfare analysis. Section~\ref{sec:discussion} discusses limitations and extensions. Section~\ref{sec:policy} draws policy implications. Section~\ref{sec:conclusion} concludes.

\section{Literature Review}
\label{sec:literature}
Housing development is subject to a coordination problem that policy alone cannot solve. Upzoning, express permitting, and deregulation may be necessary to hasten development, but they are rarely sufficient. Developers commit to large capital expenditures only when an area has already attracted consumers, which itself requires that developers have entered. The result is a 'crane' effect: each market participant waits for the others to commit first, and the market may settle below its potential.

The theoretical foundation lies in the literature on coordination games with strategic complementarities. \citet{cooper_john_1988} formalized coordination failure: when agents' optimal actions are increasing in others' actions and generate positive spillovers, multiple Pareto-ranked equilibria can arise, with the 'low' equilibrium self-fulfilling through pessimism. \citet{murphy_shleifer_vishny_1989} applied this logic to industrialization, arguing that a 'big push', a sufficiently large, coordinated investment, can shift an economy from a no-industrialization trap to coordinated industrialization. This paper transposes the big-push logic from industrial to housing development markets.

The network effects underlying this paper are empirically grounded. \citet{rossi_hansberg_sarte_owens_2010} estimate housing externalities directly using targeted revitalization programs in Richmond, Virginia, finding spillovers that decay by half every 990 feet and raise land prices in targeted neighborhoods 2–5 percent annually above controls. \citet{glaeser_kolko_saiz_2001} provide the theoretical case for increasing returns to density: density supports diverse services, public goods, low transport costs, and frequent social contact, all of which raise residents' willingness to pay. The convexity of these returns finds support in \citet{couture_handbury_2020}, who provide evidence that initial levels of non-tradable amenities better explain young-professional urbanization than changes in those levels do, consistent with convex returns to density. The multiple-equilibrium structure derived here has direct empirical counterparts. \citet{card_mas_rothstein_2008} use regression discontinuity on Census tract data to show that neighborhood racial composition exhibits sharp tipping behavior, with tipping points between five and twenty percent minority share. Their underlying bid-rent model generates the same three-equilibrium structure (two stable, one unstable) that arises in Section~\ref{sec:atomistic} from convex network effects, both validating the empirical relevance of multiplicity and providing a methodological template---detection of the unstable equilibrium via regression discontinuity---that Section~\ref{subsec:empirical} adapts for supply rather than composition dynamics.

The coordination mechanism here, irreversible investment that reshapes the subsequent game, draws on \citet{schelling_1960}'s foundational analysis of commitment as a strategic device. In standard Stackelberg models, the leader's first-mover advantage typically reduces welfare through output restriction, generating a deadweight loss that offsets any coordination gains. The welfare improvement derived here is an instance of the theory of the second best \citep{lipsey_lancaster_1956}. When network externalities already prevent the first-best outcome, concentrating commitment in a first mover can be welfare-improving rather than welfare-reducing. Sections~\ref{sec:stackelberg} and~\ref{sec:welfare} develop the specific mechanism and show how the competitive atomistic fringe disciplines the leader's markup.

Three bodies of prior work warrant direct comparison. The first is the literature on multi-sided platforms. \citet{rochet_tirole_2003, rochet_tirole_2006} and \citet{armstrong_2006} establish that platforms with cross-group externalities face coordination problems admitting multiple equilibria and characterize optimal pricing across user groups. \citet{weyl_2010} shows that a monopoly platform can use insulating tariffs (prices contingent on realized participation on each side) to implement any desired allocation, decomposing the resulting welfare loss into a classical markup distortion and a \citet{spence_1975} distortion that internalizes network externalities only to marginal users. The platform mechanism does not transfer to housing. Units sell at a single market-clearing price, so participation-contingent pricing is unavailable; physical and irreversible quantity commitment is credible in a way platform tariffs often are not; and the competitive atomistic fringe pins total supply at the equilibrium price in the regime where the fringe enters, eliminating the output-restriction incentive that drives the markup distortion in \citet{weyl_2010}. The coordination gain here therefore accrues without an accompanying markup welfare loss, distinguishing the residential quantity-commitment mechanism from the platform-pricing mechanisms in this literature.

The second near-predecessor is \citet{economides_1996}, who observes in a footnote that additional competitors' upward shift in the industry response curve can eliminate low-sales equilibria in a Stackelberg quantity-setting model with network externalities. The observation is made in passing as a corollary of his main result on the leader's incentive to invite entry. It is not accompanied by equilibrium characterization, welfare analysis, or conditions governing when elimination occurs. This paper takes the phenomenon as its central object. It characterizes the full equilibrium correspondence faced by the Stackelberg leader, proves that the leader strictly prefers commitments that eliminate the low equilibrium under general demand and cost distributions, and develops the absent welfare analysis. The closest urban-economics predecessor is \citet{helsley_strange_1994}, who model city formation with Stackelberg commitment by developers providing public goods. This paper differs in modeling positive network externalities rather than public goods, focusing on within-neighborhood coordination rather than inter-city competition, and deriving the stronger result that the leader always strictly prefers the high equilibrium. The closest residential-housing predecessor is \citet{owens_rossi_hansberg_sarte_2017}, who show in a calibrated spatial model of Detroit that government development guarantees can eliminate a vacant-neighborhood equilibrium. Their result relies on a non-market guarantor and is established computationally for a specific metro area; this paper shows that a private profit-maximizing developer endogenously plays the analogous role in a tractable applied-theory framework, and characterizes the parameter condition $\gamma \alpha S_{\text{high}}^{\alpha - 1} > 1$ under which the developer's expansion further improves welfare beyond the high equilibrium. No prior paper combines these elements---private market-driven Stackelberg commitment, residential housing with multiplicity from network externalities, and an explicit parameter condition governing when the coordination gain is accompanied by further welfare improvement---into a single model.

\section{Model}
\label{sec:model}
The preceding section situated this paper within the literatures on coordination games, housing network effects, and Stackelberg leadership. This section presents the model: a two-stage game in which a large developer commits to a quantity of housing before atomistic developers make entry decisions, with demand that depends on aggregate supply through network effects.
\subsection{Primitives}
\label{subsec:primitives}

Consider a neighborhood containing a fixed stock of potential building sites, each capable of producing one homogeneous unit of housing. Sites are heterogeneous in construction cost, reflecting differences in site characteristics. Costs are publicly observable, or sufficiently estimable from site characteristics that we treat them as such. Indexing sites by their construction cost $c$, the distribution of costs across the stock is described by a CDF $G$ with density $g$ on $[0, \bar{c}]$. The total mass of sites is normalized to 1, so any supply level $S$ satisfies $S \in [0,1]$.

There are two types of agents:

\begin{itemize}
    \item \textbf{One large developer} (the Stackelberg leader), denoted $L$. The large developer leverages scale and accumulated expertise to evaluate site profitability faster and secure construction inputs at lower cost, enabling them to commit to development before smaller competitors can respond.
    \item \textbf{A continuum of atomistic developers}, each capable of building one site. Lacking the scale to accelerate evaluation or secure inputs at favorable cost, atomistic developers cannot commit before the leader moves. They build from whatever remains of the site stock after the leader has committed.
\end{itemize}

Aggregate inverse demand for housing is $P(Q^D, S)$, where $Q^D$ is aggregate quantity demanded and $S$ is total supply. The function $P$ satisfies:

\begin{assumption}[Demand Structure]\label{as:demand}
The demand function satisfies $\frac{\dP}{\dQD} < 0$ (downward-sloping) and $\frac{\dP}{\dS} > 0$ (positive network effects: more supply raises willingness to pay).
\end{assumption}

\begin{assumption}[Regularity]\label{ass:regularity}
The cost CDF $G$ is continuously differentiable with density $g > 0$ on the interior of $[0, \bar{c}]$.
\end{assumption}

The market-clearing condition is $Q^D = S$, so the equilibrium price at total supply $S$ is $P(S,S)$. The cost distribution $G$ and the demand function $P$ are common knowledge.

\subsection{Timing and Rules of the Game}
\label{subsec:timing}

The game has two periods, $t = 0$ and $t = 1$. The rules of the game are as follows: 
\subsection*{Period $t = 0$: Large Developer's Move}
At period $t = 0$, the large developer chooses a set of sites of total mass $S_L$ to acquire and develop. Because output is homogeneous and the leader claims all revenue, profit-maximization implies they select the lowest-cost sites available, those with $c \in [0, G^{-1}(S_L)]$.\footnote{Formally, the leader chooses a measurable subset $A_L \subseteq [0, \bar{c}]$ with $G(A_L) = S_L$. Since sites are homogeneous in output and payoffs are decreasing in cost, any optimal choice is equivalent to the lowest-cost interval. The leader's action can therefore be summarized by the scalar $S_L$ without loss of generality.} The leader incurs total construction cost
\begin{equation*}
    C_L(S_L) = \int_0^{G^{-1}(S_L)} c\, g(c)\, dc.
\end{equation*}
The commitment is observable to all players before period $t = 1$. Section~\ref{subsec:alt_cost} considers an alternative environment in which the leader produces using a separate constant-cost technology.

\subsection*{Period $t = 1$: Atomistic Developers' Moves}
After observing $S_L$, each remaining site, those with cost $c \in (G^{-1}(S_L), \bar{c}]$, is assigned to a distinct atomistic developer. Atomistic developers are all price takers. Atomistic developer $j$, assigned a site with cost $c_j > G^{-1}(S_L)$, makes a binary choice: \[
a_j \in \{0,1\}
\]
where $a_j = 1$ means build and $a_j = 0$ means do not build. If developer $j$ builds, they incur cost $c_j$ and receive the market-clearing price of a unit. All atomistic developers choose simultaneously.

\begin{figure}[h!]
    \centering
    \includegraphics[width=0.8\textwidth]{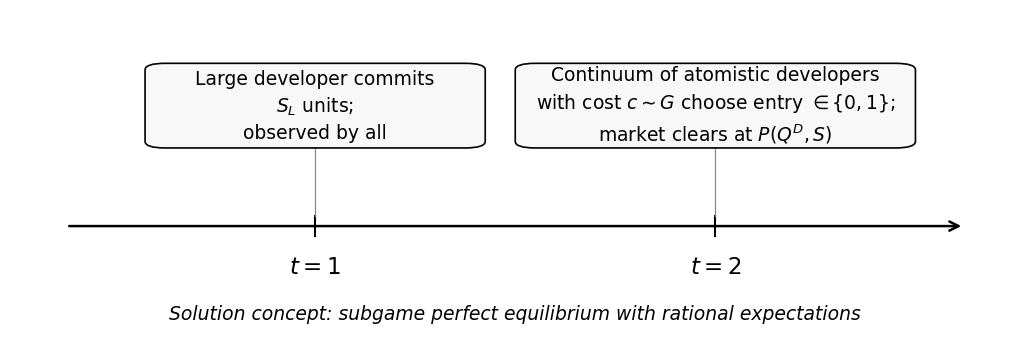}
    \caption{Timing of the game.}
    \label{fig:timeline}
\end{figure}

\subsubsection*{Market Clearing}

After both periods, the market clears. Define: \[
S_{\text{small}} = \int_{\{j: a_j = 1\}} dj
\]
as the mass of atomistic developers who chose to build in period $t = 1$. Total supply is therefore: \[
S = S_{\text{small}} + S_L
\]
All units, regardless of who built them, sell at the single market clearing price $P(S,S)$. There is no price discrimination; all units sell at $P(S,S)$. 

\subsection{Payoffs}
\label{subsec:payoffs}

For any configuration of choices $(S_L, \{a_j\}_j)$, payoffs are determined as follows.

\subsubsection*{Large Developer's Payoff}

\begin{equation}\label{eq:piL}
    \pi_L = p^* \cdot S_L - C_L(S_L) = P(S, S) \cdot S_L - \int_0^{G^{-1}(S_L)} c \, g(c) \, dc
\end{equation}

The large developer earns the market-clearing price on each of its $S_L$ units and pays the cumulative cost of the $S_L$ lowest-cost sites in the stock. Revenue depends on total supply $S$ through the price $P(S,S)$, but costs depend only on the developer's quantity $S_L$.

\subsubsection*{Atomistic Developer $j$'s Payoff}

\begin{equation}\label{eq:pij}
    \pi_j = \begin{cases} p^* - c_j = P(S, S) - c_j & \text{if } a_j = 1 \\ 0 & \text{if } a_j = 0 \end{cases}
\end{equation}

An atomistic developer who builds earns the difference between the market price and their idiosyncratic cost. A developer who does not build earns zero, regardless of market conditions.

\subsubsection*{Consumer Welfare}

\begin{equation}\label{eq:CS}
    CS = \int_0^{S} P(Q^D, S) \, dQ^D - P(S,S) \cdot S
\end{equation}

Consumer welfare is the area under the demand curve up to total supply $S$, minus total expenditure. The network effect enters through the second argument of $P$. Holding quantity fixed, higher aggregate supply raises willingness to pay.
\subsubsection*{Total Welfare}

\begin{equation}\label{eq:welfare_raw}
    W = \int_0^{S} P(Q^D, S) \, dQ^D - \int_0^{G^{-1}(S_L)} c \, g(c) \, dc - \int_{\{j : a_j = 1\}} c_j \, dj
\end{equation}

In equilibrium, the large developer builds units in $[0, G^{-1}(S_L)]$ and atomistic entrants build the remaining lowest-cost units in $(G^{-1}(S_L), G^{-1}(S)]$. Therefore total production cost at supply $S$ is $\int_0^{G^{-1}(S)} c \, g(c) \, dc$ regardless of the split between the large and atomistic developers. Hence, for equilibrium allocations, welfare can be written as:
\begin{equation}\label{eq:welfare}
    W(S) = \int_0^{S} P(Q^D, S) \, dQ^D - \int_0^{G^{-1}(S)} c \, g(c) \, dc
\end{equation}

Prices do not appear in the welfare expression because total price paid, $P(S,S)\cdot S$, is subtracted from gross consumer benefit and added to producer surplus, so it drops out. The reduced-form expression in \eqref{eq:welfare} applies whenever the realized allocation builds the lowest-cost mass $S$ of units; in particular, it applies at equilibrium in the baseline model and in the planner's problem in Section~\ref{subsec:planner}.

\subsection{Outcome Mapping: From Choices to Payoffs}
\label{subsec:outcome_mapping}

The previous section defined each agent's payoff as a function of total supply $S$ and the market-clearing price $P(S,S)$. However, these objects are not chosen directly but are instead determined by the underlying actions $(S_L, \{a_j\}_j)$. This section makes the connection from any action configuration to aggregate quantities, prices, and payoffs. Given any $(S_L, \{a_j\}_j)$:

\begin{enumerate}
    \item Compute $\Ssmall = \int_{\{j : a_j = 1\}} dj$.
    \item Compute $S = S_L + \Ssmall$.
    \item Compute $p^* = P(S, S)$.
    \item Compute $\pi_L = p^* \cdot S_L - \int_0^{G^{-1}(S_L)} c \, g(c) \, dc$.
    \item For each atomistic developer $j$: $\pi_j = a_j \cdot (p^* - c_j)$.
    \item Compute $W = \int_0^{S} P(Q^D, S) \, dQ^D - \int_0^{G^{-1}(S_L)} c \, g(c) \, dc - \int_{\{j : a_j = 1\}} c_j \, dj$.
\end{enumerate}

No element of this mapping depends on beliefs, expectations, or strategic reasoning. It is a mechanical function from actions to outcomes. In equilibrium, the welfare expression in Step~6 simplifies to Equation~\eqref{eq:welfare}, since the realized allocation builds the lowest-cost mass $S$ of units.

\subsection{Specified Demand Function}
\label{subsec:specified_demand}

The results in Sections~\ref{sec:atomistic}--\ref{sec:stackelberg} hold for 
general inverse demand $P(Q^D, S)$ satisfying Assumption~\ref{as:demand}. 
To develop sharper comparative statics and closed-form conditions, we 
will occasionally specialize to:
\begin{equation}\label{eq:specified_demand}
    P(Q^D, S) = \frac{Q_{\max} + \gamma S^\alpha - Q^D}{\beta}
\end{equation}
where $Q_{\max} > 0$ is the demand intercept, $\beta > 0$ governs price 
sensitivity, $\gamma > 0$ captures the strength of network effects, and 
$\alpha \geq 1$ governs the convexity of network benefits. This 
specification satisfies Assumption~\ref{as:demand}: 
$\frac{\partial P}{\partial Q^D} = -\frac{1}{\beta} < 0$ and 
$\frac{\partial P}{\partial S} = \frac{\gamma \alpha S^{\alpha-1}}{\beta} > 0$.

The parameter $\alpha$ plays a key role. When $\alpha = 1$, network 
benefits are linear in supply: each additional unit contributes a 
constant $\gamma/\beta$ to willingness to pay. When $\alpha > 1$, 
network benefits are convex and the marginal contribution of density 
is increasing in $S$. As shown in 
Section~\ref{subsec:multiplicity}, $\alpha > 1$ is the condition 
that generates the S-shaped response curve and hence multiple equilibria.

\section{Atomistic Equilibrium}
\label{sec:atomistic}
 
The previous section defined the game's primitives, timing, rules, and payoffs. This section characterizes the Stage 2 equilibrium concept, the equilibrium behavior of atomistic developers given an arbitrary commitment of $S_L$ by the large developer. We show that when network effects of demand are sufficiently convex, the Stage 2 game admits multiple stable equilibria, creating a coordination problem that the model's primitives alone cannot resolve. 
\subsection{Stage 2 Equilibrium (Period \texorpdfstring{$t = 1$}{t = 1})}
\label{subsec:stage2}

Fix $S_L$. A \textbf{Stage 2 Nash equilibrium} is a profile $\{a_j^*\}_j$ such that each atomistic developer best-responds to aggregate behavior:
\begin{equation}\label{eq:best_response}
    a_j^* = 1 \iff c_j \leq P(S^*, S^*)
\end{equation}
where $S^* = S_L + \Ssmall^*$ and $\Ssmall^* = \int_{\{j : a_j^* = 1\}} dj$.
 
Equivalently, $\Ssmall^*$ is a fixed point of the map
\begin{equation}\label{eq:fixed_point_small}
    f(\Ssmall) = G\big(P(S_L + \Ssmall, \, S_L + \Ssmall)\big) - S_L
\end{equation}
subject to $\Ssmall \geq 0$. For interior equilibria where atomistic developers enter ($S_{\text{small}}^* > 0$), total equilibrium supply $S^*$ satisfies:
\begin{equation}\label{eq:fixed_point}
    S^* = G\big(P(S^*, S^*)\big) \quad \text{subject to} \quad S^* \geq S_L
\end{equation}
This is the same fixed-point condition as the pure atomistic game, $P(S,S) = G^{-1}(S)$, restricted to $S \geq S_L$. Corner equilibria, where $S^* = S_L$, which arise when $G(P(S_L, S_L)) < S_L$, are not captured by this equality and are handled by the more general formulation in Section~\ref{subsec:multiplicity}.
 
\subsection{Multiple Equilibria}
\label{subsec:multiplicity}
 
When network effects are sufficiently strong ($\frac{\dP}{\dS}$ sufficiently large relative to the demand slope $\frac{\dP}{\dQD}$), the fixed-point condition $S = G(P(S,S))$ admits multiple solutions. Denote the stable equilibria by:
\begin{equation}
    \Slow < \Shigh
\end{equation}
with an unstable equilibrium $\Sunstable$ between them.

Throughout, ``equilibrium'' in the Stage 2 continuation game refers to stable equilibrium unless otherwise stated. The unstable fixed point $\Sunstable$ is a threshold separating basins of attraction, not a predicted market outcome.

\begin{figure}[H]
    \centering
    \includegraphics[width=0.5\textwidth]{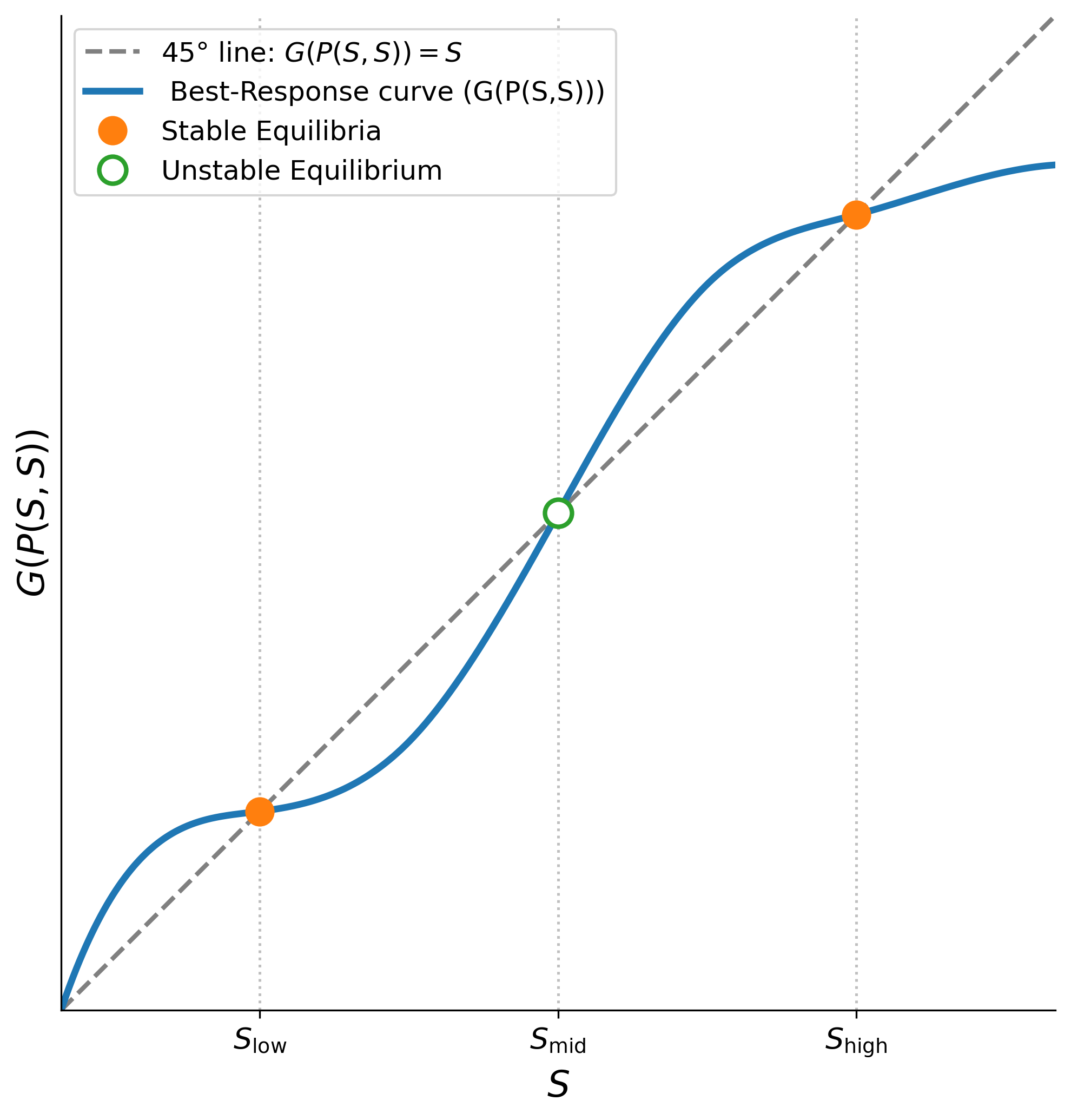}
    \caption{Plot of the response curve $G(P(S,S))$ against the 45-degree line. Intersections are equilibria. At $S_{\text{low}}$ and $S_{\text{high}}$, the curve crosses from above, so perturbations are self-correcting and these equilibria are stable. At $S_{\text{mid}}$, the curve crosses from below so perturbations are self-reinforcing and this equilibrium is unstable.}
    \label{fig:fixed_point}
\end{figure}
 
A Stage 2 equilibrium given $S_L$ is a total supply level $S^*$ satisfying:
\begin{equation}
S^* = \max\bigl(S_L,\; G(P(S^*, S^*))\bigr)
\end{equation}
This nests both interior equilibria, where $S^* = G(P(S^*,S^*)) \geq S_L$, and a corner equilibrium at $S^* = S_L$ with $S_{\text{small}} = 0$ whenever $G(P(S_L, S_L)) < S_L$. The set of stable Stage 2 equilibria given $S_L$ is:
\begin{equation}\label{eq:equilibrium_set}
\mathcal{E}(S_L) = \{S^* : S^* = \max(S_L,\; G(P(S^*, S^*))),\; S^* \text{ is stable}\}
\end{equation}

Four cases arise:

\begin{tabular}{lll}
\toprule
Condition & $\mathcal{E}(S_L)$ & Interpretation \\
\midrule
$S_L \leq S_{\text{low}}$ & $\{S_{\text{low}},\, S_{\text{high}}\}$ & Both interior equilibria feasible \\
$S_{\text{low}} < S_L < S_{\text{unstable}}$ & $\{S_L,\, S_{\text{high}}\}$ & Corner trap and high equilibrium \\
$S_{\text{unstable}} \leq S_L \leq S_{\text{high}}$ & $\{S_{\text{high}}\}$ & Coordination failure eliminated \\
$S_L > S_{\text{high}}$ & $\{S_L\}$ & Monopoly region: $S_{\text{small}} = 0$ \\
\bottomrule
\end{tabular}

In the second case, the response curve $G(P(S,S))$ lies below the 45-degree line at $S = S_L$, so $P(S_L, S_L) < G^{-1}(S_L)$. The market price is below the cost threshold of the cheapest remaining atomistic developer. Zero atomistic entry is therefore self-confirming, and $S^* = S_L$ is a stable corner equilibrium. The coordination failure takes a different form than in Regime~I. The low interior equilibrium has been eliminated, but a no-entry trap persists, and is not resolved until $S_L$ reaches $S_{\text{unstable}}$.

\begin{figure}[htbp]
    \centering
    \includegraphics[width=0.75\textwidth]{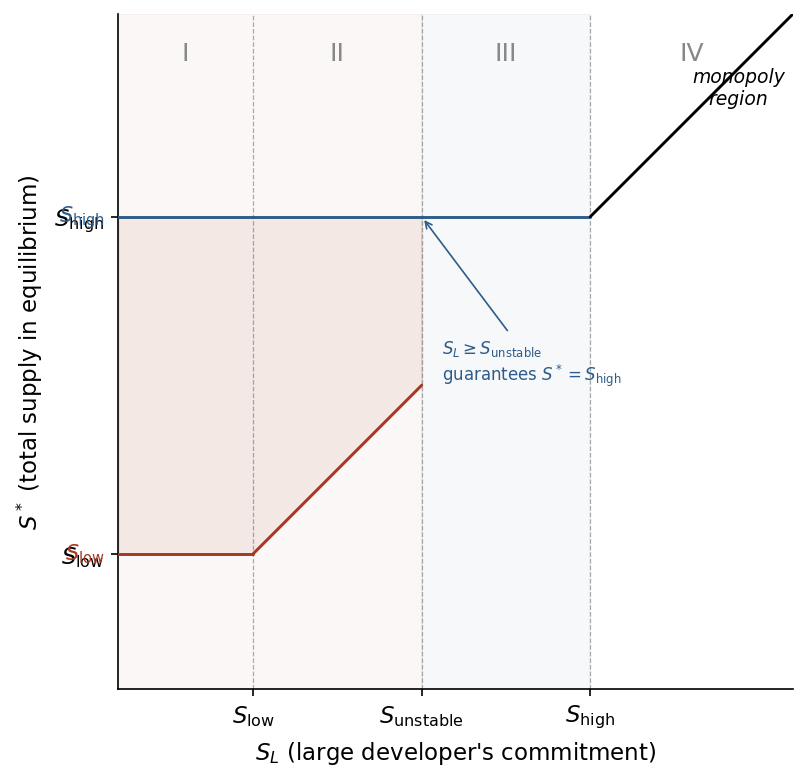}
    \caption{Equilibrium correspondence as a function of the large developer’s commitment $S_L$. For $S_L \le S_{low}$ and $S_{low} < S_L < S_{unstable}$, multiple equilibria exist, including a low equilibrium, a high equilibrium, and (in the intermediate region) a corner equilibrium at $S^* = S_L$. For $S_L \ge S_{unstable}$, the correspondence becomes single-valued with $S^* = S_{high}$. For $S_L > S_{high}$, atomistic entry ceases and the market enters the monopoly region with $S^* = S_L$.}
    \label{fig:regime_partition}
\end{figure}

Figure~\ref{fig:regime_partition} illustrates the equilibrium correspondence across regimes.

\begin{proposition}[Equilibrium Stability]\label{prop:stability}
An equilibrium $S^*$ is stable if and only if $\frac{d}{dS}G(P(S,S))\big|_{S=S^*} < 1$. When the response curve $G(P(S,S))$ is S-shaped (which arises here from sufficient convexity in network effects), the low and high equilibria are stable and the intermediate equilibrium is unstable.
\end{proposition}
 
\begin{proof}
An equilibrium $S^*$ is stable if small perturbations generate restoring forces. Consider $S$ slightly below $S^*$. If $G(P(S,S)) > S$, then more developers wish to build than currently do, creating upward pressure. If $G(P(S,S)) < S$, fewer developers wish to build, creating downward pressure. Stability requires that deviations below equilibrium generate upward pressure and deviations above generate downward pressure.
 
This occurs when the response curve $G(P(S,S))$ crosses the 45-degree line from above, which requires:
\begin{equation}
    \frac{d}{dS} G(P(S,S)) < 1
\end{equation}
 
Computing this derivative:
\begin{equation}
    \frac{d}{dS} G(P(S,S)) = g(P(S,S)) \cdot \left[\frac{\dP}{\dQD}\bigg|_{(S,S)} + \frac{\dP}{\dS}\bigg|_{(S,S)}\right]
\end{equation}
 
When the network effect $\frac{\dP}{\dS}$ is sufficiently large relative to the demand slope $\left|\frac{\dP}{\dQD}\right|$, this derivative exceeds 1 in a middle region, generating the S-shaped response curve with crossings from below (unstable) flanked by crossings from above (stable) at $\Slow$ and $\Shigh$.
\end{proof}

For the remainder of the paper, we maintain the following condition.

\begin{assumption}[Multiplicity]\label{ass:multiplicity}
The composite map $S \mapsto G(P(S,S))$ admits exactly two stable fixed points 
$S_{\text{low}} < S_{\text{high}}$ separated by exactly one unstable fixed point $S_{\text{unstable}}$, with no further crossings of the 45-degree line.
\end{assumption}

This is the case in which the coordination problem is non-trivial. Section~\ref{subsec:specified_demand} provides sufficient conditions under the specified demand function: $\alpha > 1$ with $\gamma$ sufficiently large relative to $\beta$ and the cost density $g$.

\begin{remark}[Economic interpretation of convexity]
The condition $\alpha > 1$ has a direct empirical interpretation. Neighborhood amenities (retail viability, transit service, local services) exhibit thresholds beyond which density generates disproportionately large gains. Consistent with the evidence in \citet{couture_handbury_2020}, the level of amenity density, rather than incremental changes in it, drives the attractiveness of urban neighborhoods. The S-shaped response curve and the resulting multiplicity are the formal expression of this threshold structure. Below a critical density, the marginal benefit of an additional unit is too small to sustain entry and above it, each new unit further raises willingness to pay, reinforcing entry.
\end{remark}

\subsection{The Indeterminacy Problem}
\label{subsec:indeterminacy}
When $\mathcal{E}(S_L)$ contains multiple elements, the model does not predict which equilibrium will materialize. Atomistic developers face a pure coordination problem: each developer's optimal action depends on what they expect others to do, and the model provides no mechanism to select among the self-consistent outcomes.

This indeterminacy is the central problem the paper addresses. When $S_L \leq S_{\text{low}}$, the market may settle at $S_{\text{low}}$ or $S_{\text{high}}$. When $S_L \in (S_{\text{low}}, S_{\text{unstable}})$, the low interior equilibrium has been eliminated, but a corner trap at $S^* = S_L$ coexists with $S_{\text{high}}$. The developer's commitment has raised the floor but not yet guaranteed coordination on the high equilibrium. In both ranges, the developer's payoff is uncertain, driven by unmodeled factors such as developer sentiment, historical precedent, or sheer luck. The equilibrium supply correspondence given $S_L$ is:
\begin{equation}\label{eq:supply_mapping}
S^*(S_L) \in \begin{cases}
\{S_{\text{low}},\, S_{\text{high}}\} & \text{if } S_L \leq S_{\text{low}} \\
\{S_L,\, S_{\text{high}}\} & \text{if } S_{\text{low}} < S_L < S_{\text{unstable}} \\
\{S_{\text{high}}\} & \text{if } S_{\text{unstable}} \leq S_L \leq S_{\text{high}} \\
\{S_L\} & \text{if } S_L > S_{\text{high}}
\end{cases}
\end{equation}

The key observation is that once $S_L \geq S_{\text{unstable}}$, the indeterminacy vanishes. For $S_L \in [S_{\text{unstable}}, S_{\text{high}}] $, the unique Stage 2 equilibrium is $S_{\text{high}}$, and for $S_L > S_{\text{high}}$, the unique equilibrium is $S_L$ itself. The large developer's commitment to building past the unstable equilibrium threshold eliminates all coordination failures and guarantees an outcome of at least $S_{\text{high}}$.
 
 
\section{Stackelberg Equilibrium}
\label{sec:stackelberg}
 
Section~\ref{sec:atomistic} characterized the Stage 2 equilibrium structure and showed that when network effects are sufficiently convex, the atomistic market admits multiple stable equilibria. This section defines equilibrium for the full two-stage game, states the main results on the large developer's behavior, and then proves them.
 
\subsection{Equilibrium Definition}
\label{subsec:equil_def}

The large developer moves before atomistic developers and must form beliefs about how Stage 2 play will respond to each possible choice of $S_L$. In equilibrium, these beliefs are correct: the developer's conjecture about the mapping from $S_L$ to Stage 2 outcomes coincides with the actual equilibrium correspondence derived in Section~\ref{sec:atomistic}.

\begin{definition}
An equilibrium of the Stackelberg game is a pair $(S_L^*, \sigma^*)$, where $S_L^* \geq 0$ is the large developer's Stage 1 quantity and $\sigma^* = \{a_j^*\}_j$ is a Stage 2 strategy profile, such that:

\begin{enumerate}
\item[(i)] Given $S_L^*$, the profile $\sigma^*$ constitutes a Stage 2 Nash equilibrium inducing total supply $S^* \in \mathcal{E}(S_L^*)$. Each atomistic developer builds if and only if $c_j \leq P(S^*, S^*)$, where $S^* = S_L^* + S_{\text{small}}^*$.

\item[(ii)] For every deviation $S_L \neq S_L^*$ and every continuation $S' \in \mathcal{E}(S_L)$,
\[
\pi_L(S_L^*, S^*) \geq \pi_L(S_L, S').
\]
\end{enumerate}
\end{definition}
 
Condition~(ii) ensures that the developer's choice is robust to the indeterminacy in Stage~2. By definition, the continuation correspondence $\mathcal{E}(S_L)$ includes only stable Stage~2 equilibria. The remaining indeterminacy is therefore among stable continuation equilibria only. Because the developer cannot control which element of $\mathcal{E}(S_L)$ materializes, the dominance criterion requires that $S_L^*$ outperform every alternative even under the most favorable Stage~2 outcome for the deviation. Thus, the main result is best interpreted as a robustness result rather than as an equilibrium-selection result.

\subsection{Main Results}
\label{subsec:main_results}
 
Given any choice of $S_L$, the large developer's profit depends on which 
Stage 2 equilibrium materializes. The equilibrium correspondence 
$\mathcal{E}(S_L)$ is multi-valued when $S_L < S_{\text{unstable}}$ but 
single-valued for $S_L \geq S_{\text{unstable}}$. Our main result shows 
that the developer's optimal choice lies in Regime~II or~III, where the 
correspondence is single-valued, and that this choice dominates every 
alternative $S_L < S_{\text{unstable}}$ under every possible Stage~2 
realization. The dominance criterion in Definition~1 is therefore 
satisfied without imposing any equilibrium selection assumption on the 
Stage~2 game. The developer's profit given $S_L$ and Stage~2 outcome 
$S^* \in \mathcal{E}(S_L)$ is:
\begin{equation}\label{eq:developer_profit}
    \pi_L(S_L) = P(S^*, S^*) \cdot S_L - \int_0^{G^{-1}(S_L)} c \, g(c) \, dc
\end{equation}
where $S^* \in \mathcal{E}(S_L)$. The structure of $\mathcal{E}(S_L)$ from equation~\eqref{eq:supply_mapping} partitions the large developer's choice set into four regimes, analyzed in the proof below. The following proposition characterizes the equilibrium.
 
\begin{proposition}[Equilibrium Characterization]\label{prop:main}
\leavevmode
\begin{enumerate}
    \item[(a)] In any equilibrium, $S_L^* \geq \Shigh$. The large developer always commits at least to the high equilibrium level, which requires pushing past the unstable threshold $\Sunstable$. This eliminates the coordination failure and guarantees total supply of at least $\Shigh$ regardless of which Stage 2 equilibrium would have obtained under any alternative $S_L < \Sunstable$; no equilibrium selection assumption is needed.
    \item[(b)] The large developer's optimal choice is either $S_L^* = \Shigh$ (if high-equilibrium profit dominates monopoly profit) or $S_L^* = S_L^{\text{mon}} > \Shigh$ (if monopoly profit dominates), where $S_L^{\text{mon}}$ satisfies the first-order condition~\eqref{eq:regime3_foc}.
    \item[(c)] The large developer finds it strictly profitable to build past $S_{\text{high}}$ into the monopoly region if and only if network effects dominate the demand curve slope at $S_{\text{high}}$:
    \begin{equation}
        \frac{\partial P}{\partial Q^D}\bigg|_{(S_{\text{high}}, S_{\text{high}})} + \frac{\partial P}{\partial S}\bigg|_{(S_{\text{high}}, S_{\text{high}})} > 0
    \end{equation}
\end{enumerate}
\end{proposition}
 
In Regime~II and in Regime~I under the $S_{\text{low}}$ realization, the 
large developer optimally builds the entire equilibrium supply, 
displacing all atomistic developers. The total profit when building 
the entire supply at equilibrium $S$, $\Pi(S)$, is strictly increasing 
because a higher equilibrium supply corresponds to a higher equilibrium 
price $G^{-1}(S)$, which increases revenue on all inframarginal units 
whose costs lie strictly below $G^{-1}(S)$. Thus, the large developer always prefers the high equilibrium, and committing past the unstable threshold to trigger $\Shigh$ is always profitable. A large developer with first-mover commitment power solves the coordination problem and guarantees that the market settles at minimum at the high equilibrium.

\begin{figure}[H]
    \centering
    \includegraphics[width=0.7\textwidth]{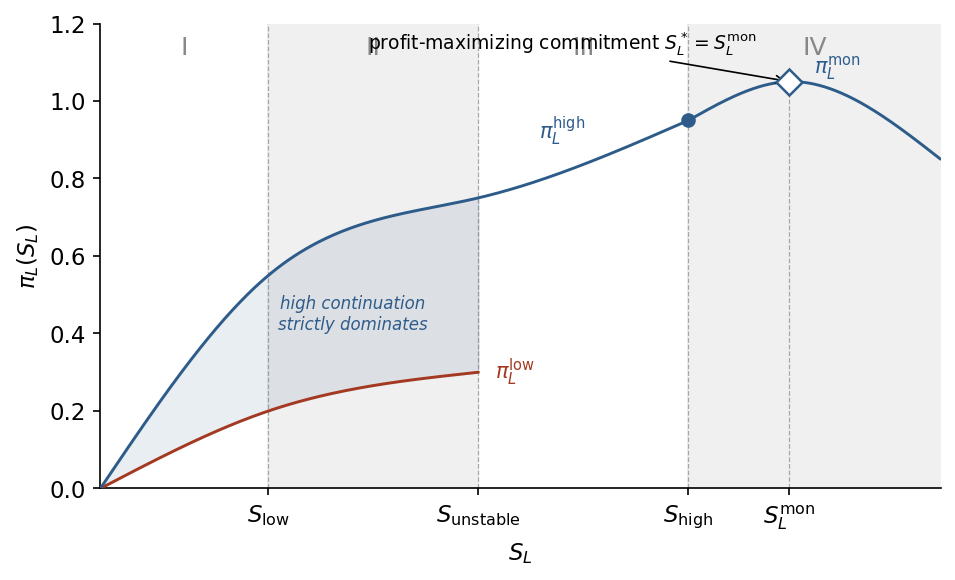}
    \caption{Plot of large developer profit against $S_L$ in high (blue) and low/corner (red) equilibrium continuation. High equilibrium continuation strictly dominates low equilibrium continuation for all $S_L < S_{\text{unstable}}$. Figure illustrates sufficiently strong and convex network effects, with $S^*_L = S_L^{\text{mon}}$. }
    \label{fig:profit_dominance}
\end{figure}

\begin{remark}[Application to Specified Demand]
Under $P(Q^D, S) = \frac{Q_{\max} + \gamma S^\alpha - Q^D}{\beta}$, the condition in part (c) becomes $\gamma \alpha (\Shigh)^{\alpha - 1} > 1$. The large developer's profit is locally increasing past $\Shigh$ (thus causing the large developer to build past $S_{\text{high}}$) when network effects are sufficiently strong (large $\gamma$), returns to density are sufficiently convex (high $\alpha$), or the high equilibrium supply level is sufficiently high. When these conditions fail, the demand curve slope dominates and the large developer does not exceed $\Shigh$.
\end{remark}
 
\subsection{Proof of Proposition~\ref{prop:main}}
\label{subsec:proof}
 
The proof proceeds by analyzing the developer's profit across four regions of $S_L$ determined by the equilibrium correspondence $\mathcal{E}(S_L)$.
 
\paragraph{Case 1: $S_L \leq \Slow$ (Regime I).}
 
When $S_L \leq \Slow$, both $\Slow$ and $\Shigh$ are feasible Stage 2 equilibria and the large developer cannot control which one materializes. Total supply is either $\Slow$ or $\Shigh$, and the large developer's profit is correspondingly either:
\begin{equation}
    \pi_L^{I, \text{low}}(S_L) = G^{-1}(\Slow) \cdot S_L - \int_0^{G^{-1}(S_L)} c \, g(c) \, dc
\end{equation}
or
\begin{equation}
    \pi_L^{I, \text{high}}(S_L) = G^{-1}(\Shigh) \cdot S_L - \int_0^{G^{-1}(S_L)} c \, g(c) \, dc
\end{equation}
 
In either case, the market-clearing price is constant in $S_L$. The large developer replaces atomistic developers one-for-one, leaving total supply unchanged at each fixed point. Both expressions are therefore strictly increasing in $S_L$ on $[0, \Slow)$, since each additional unit costs $G^{-1}(S_L)$, which is below either equilibrium price. The optimal choice in Regime I is therefore $S_L^I = \Slow$, with maximized profit:
\begin{equation}\label{eq:piI_star}
    \pi_L^{I*} = \begin{cases}
        G^{-1}(\Slow) \cdot \Slow - \displaystyle\int_0^{G^{-1}(\Slow)} c \, g(c) \, dc & \text{if market coordinates on } \Slow \\[10pt]
        G^{-1}(\Shigh) \cdot \Slow - \displaystyle\int_0^{G^{-1}(\Slow)} c \, g(c) \, dc & \text{if market coordinates on } \Shigh
    \end{cases}
\end{equation}
 
Even in the best case (market coordinates on $\Shigh$), the developer's profit is bounded above by $\pi_L^{I, \text{high}}(\Slow) = G^{-1}(\Shigh) \cdot \Slow - \int_0^{G^{-1}(\Slow)} c \, g(c) \, dc$, since Regime I constrains $S_L \leq \Slow$.
 
\paragraph{Case 2a: $\Slow < S_L < S_{\text{unstable}}$ (Intermediate Regime).}

Once $S_L$ exceeds $\Slow$, the low interior equilibrium disappears, but a corner equilibrium at $S^* = S_L$ with $S_{\text{small}} = 0$ persists alongside $\Shigh$. In this range, $G(P(S_L, S_L)) < S_L$, so the market price $P(S_L, S_L)$ is below $G^{-1}(S_L)$: no remaining atomistic developer can profitably enter, sustaining zero entry as a self-confirming equilibrium.

Under the corner realization, the developer's profit is:
\begin{equation}
    \pi_L^{\text{corner}}(S_L) = P(S_L, S_L) \cdot S_L - \int_0^{G^{-1}(S_L)} c \, g(c) \, dc
\end{equation}
Since $P(S_L, S_L) < G^{-1}(S_L)$ in this region:
\[
    \pi_L^{\text{corner}}(S_L) < G^{-1}(S_L) \cdot S_L - \int_0^{G^{-1}(S_L)} c \, g(c) \, dc = \Pi(S_L) < \Pi(\Shigh) = \pi_L^{II*}
\]
where the last inequality follows from $\Pi$ being strictly increasing (proved below in the cross-regime comparison). Under the $\Shigh$ realization, the developer's profit is $G^{-1}(\Shigh) \cdot S_L - \int_0^{G^{-1}(S_L)} c \, g(c) \, dc < \pi_L^{II*}$ since $S_L < \Shigh$. The developer therefore strictly prefers any Regime~II choice to any choice in the intermediate regime under either Stage~2 realization.

\paragraph{Case 2b: $S_{\text{unstable}} \leq S_L \leq \Shigh$ (Regime II).}

Once $S_L$ reaches $S_{\text{unstable}}$, the corner equilibrium is also eliminated and the unique Stage 2 equilibrium is $\Shigh$. The indeterminacy is fully resolved: the large developer's commitment guarantees the high-supply outcome. Total supply is $\Shigh$ and the price of a unit is $G^{-1}(\Shigh)$, constant in $S_L$. The same displacement logic applies; the large developer replaces atomistic developers at a fixed price in this regime, with profit given by:
\begin{equation}
    \pi_L^{II}(S_L) = G^{-1}(\Shigh) \cdot S_L - \int_0^{G^{-1}(S_L)} c \, g(c) \, dc
\end{equation}
 
Differentiating with respect to $S_L$:
\begin{equation}
    \frac{d\pi_L^{II}}{dS_L} = G^{-1}(\Shigh) - G^{-1}(S_L)
\end{equation}
 
Since $S_L \leq \Shigh$ implies $G^{-1}(S_L) \leq G^{-1}(\Shigh)$, profit is strictly increasing on $(S_{\text{unstable}}, \Shigh)$. Thus, the optimal choice of $S_L$ in Regime II is $S_L^{II} = \Shigh$, yielding maximized profit:
\begin{equation}\label{eq:piII_star}
    \pi_L^{II*} = G^{-1}(\Shigh) \cdot \Shigh - \int_0^{G^{-1}(\Shigh)} c \, g(c) \, dc
\end{equation}
 
\paragraph{Case 3: $S_L > \Shigh$ (Regime III).}
 
Since $\Shigh$ is the largest stable fixed-point solution to $S = G(P(S,S))$, we have $S > G(P(S,S))$ for any $S > \Shigh$ (because the curve $G(P(S,S))$ crosses the 45-degree line from above at $\Shigh$). Thus, when $S_L > \Shigh$, $G^{-1}(S_L) > P(S_L, S_L)$ so the cost of the next unit is higher than the price it will fetch on the market. No atomistic developers will enter, total supply is given by $S_L$, and the market becomes a true monopolist environment. The large developer therefore chooses $S_L$, subject to $S_L > \Shigh$, that maximizes:
\begin{equation}
    \pi_L^{III}(S_L) = P(S_L, S_L) \cdot S_L - \int_0^{G^{-1}(S_L)} c \, g(c) \, dc
\end{equation}
 
Differentiating:
\begin{equation}\label{eq:regime3_deriv}
    \frac{d\pi_L^{III}}{dS_L} = P(S_L, S_L) + \left[\frac{\dP}{\dQD}\bigg|_{(S_L, S_L)} + \frac{\dP}{\dS}\bigg|_{(S_L, S_L)}\right] \cdot S_L - G^{-1}(S_L)
\end{equation}
 
Unlike the preceding cases, here the large developer faces the standard monopolist tradeoff. Every extra unit it builds earns $P(S_L, S_L)$ but impacts the price by $\frac{\dP}{\dQD}\big|_{(S_L, S_L)} + \frac{\dP}{\dS}\big|_{(S_L, S_L)}$ per unit, and hence total revenue from the $S_L$ existing units by $\left[\frac{\dP}{\dQD}\big|_{(S_L, S_L)} + \frac{\dP}{\dS}\big|_{(S_L, S_L)}\right] \cdot S_L$. This can be negative when the demand curve slope dominates the network effect (and positive when vice versa). The large developer's optimal supply $S_L^{\text{mon}}$ satisfies the FOC:
\begin{equation}\label{eq:regime3_foc}
    P(S_L^{\text{mon}}, S_L^{\text{mon}}) + \left[\frac{\dP}{\dQD}\bigg|_{(S_L^{\text{mon}}, S_L^{\text{mon}})} + \frac{\dP}{\dS}\bigg|_{(S_L^{\text{mon}}, S_L^{\text{mon}})}\right] \cdot S_L^{\text{mon}} = G^{-1}(S_L^{\text{mon}})
\end{equation}
 
Maximized profit: $\pi_L^{III*} = \pi_L^{III}(S_L^{\text{mon}})$.
 
\paragraph{Cross-regime comparison: proof of part (a).}
 
We show that Regime II profit exceeds profit under every alternative $S_L < S_{\text{unstable}}$, under every possible Stage~2 realization.
 
\textit{Regime II vs.\ Regime I with $\Slow$ realization.}
Define:
\begin{equation}\label{eq:Pi_function}
    \Pi(S) \equiv G^{-1}(S) \cdot S - \int_0^{G^{-1}(S)} c \, g(c) \, dc
\end{equation}
as the profit the large developer earns when building all $S$ units of an equilibrium at price $G^{-1}(S)$. The Regime I profit under the $\Slow$ realization is $\pi_L^{I, \text{low}}(\Slow) = \Pi(\Slow)$, and the Regime II profit is $\pi_L^{II*} = \Pi(\Shigh)$.
 
Differentiating the revenue term:
\[
\frac{d}{dS}\left[G^{-1}(S) \cdot S\right] = \frac{S}{g(G^{-1}(S))} + G^{-1}(S)
\]
where $\frac{d}{dS}[G^{-1}(S)] = \frac{1}{g(G^{-1}(S))}$ by the inverse function theorem. For the cost integral, applying Leibniz's rule:
\[
\frac{d}{dS} \int_0^{G^{-1}(S)} c \, g(c) \, dc = G^{-1}(S) \cdot g(G^{-1}(S)) \cdot \frac{1}{g(G^{-1}(S))} = G^{-1}(S)
\]
Subtracting:
\begin{equation}\label{eq:Pi_prime}
    \Pi'(S) = \frac{S}{g(G^{-1}(S))} + G^{-1}(S) - G^{-1}(S) = \frac{S}{g(G^{-1}(S))}
\end{equation}
Since $S > 0$ and $g > 0$ on the interior of the support, $\Pi'(S) > 0$ for all $S > 0$. Therefore:
\[
\Shigh > \Slow \implies \Pi(\Shigh) > \Pi(\Slow) \implies \pi_L^{II*} > \pi_L^{I, \text{low}}(\Slow)
\]
 
\textit{Regime II vs.\ Regime I with $\Shigh$ realization.} If the market coordinates on $\Shigh$ in Regime I, the developer's best Regime I profit is:
\[
\pi_L^{I, \text{high}}(\Slow) = G^{-1}(\Shigh) \cdot \Slow - \int_0^{G^{-1}(\Slow)} c \, g(c) \, dc
\]
The Regime II profit is:
\[
\pi_L^{II*} = G^{-1}(\Shigh) \cdot \Shigh - \int_0^{G^{-1}(\Shigh)} c \, g(c) \, dc
\]
Taking the difference:
\[
\pi_L^{II*} - \pi_L^{I, \text{high}}(\Slow) = G^{-1}(\Shigh) \cdot (\Shigh - \Slow) - \int_{G^{-1}(\Slow)}^{G^{-1}(\Shigh)} c \, g(c) \, dc
\]
Every unit with cost $c \in [G^{-1}(\Slow), G^{-1}(\Shigh)]$ satisfies $c \leq G^{-1}(\Shigh)$, with strict inequality on a set of positive measure (since $g > 0$ on the interior of the support). Therefore:
\[
\int_{G^{-1}(\Slow)}^{G^{-1}(\Shigh)} c \, g(c) \, dc < G^{-1}(\Shigh) \cdot \int_{G^{-1}(\Slow)}^{G^{-1}(\Shigh)} g(c) \, dc = G^{-1}(\Shigh) \cdot (\Shigh - \Slow)
\]
where the final equality uses $\int_{G^{-1}(\Slow)}^{G^{-1}(\Shigh)} g(c) \, dc = G(G^{-1}(\Shigh)) - G(G^{-1}(\Slow)) = \Shigh - \Slow$. The difference is therefore strictly positive:
\[
\pi_L^{II*} > \pi_L^{I, \text{high}}(\Slow)
\]

\textit{Regime II vs.\ intermediate regime.} For any $S_L \in (\Slow, S_{\text{unstable}})$, the equilibrium correspondence contains $\{S_L, \Shigh\}$. Under the $\Shigh$ realization, the developer's profit is $G^{-1}(\Shigh) \cdot S_L - \int_0^{G^{-1}(S_L)} c\, g(c)\, dc < \pi_L^{II*}$ since $S_L < \Shigh$. Under the corner realization $S^* = S_L$, the developer's profit is $\pi_L^{\text{corner}}(S_L) = P(S_L, S_L) \cdot S_L - \int_0^{G^{-1}(S_L)} c\, g(c)\, dc$. Since $P(S_L, S_L) < G^{-1}(S_L)$ in this region, $\pi_L^{\text{corner}}(S_L) < \Pi(S_L) < \Pi(\Shigh) = \pi_L^{II*}$.
 
Since $\pi_L^{II*}$ exceeds profit under every alternative $S_L < S_{\text{unstable}}$ and every possible Stage 2 realization, the large developer strictly prefers Regime II or III. Combined with the optimum in Regime II being $S_L = \Shigh$ (Case 2b) and in Regime III being $S^{\text{mon}}_L > \Shigh$ (Case 3), we have $S^*_L \geq \Shigh$, establishing part (a).
 
\paragraph{Proof of part (c).}
 
Evaluate Regime III profit at $S_L = \Shigh$. At this point, $P(\Shigh, \Shigh) = G^{-1}(\Shigh)$ since $\Shigh$ is a fixed point. So:
\[
\frac{d\pi_L^{III}}{dS_L}\bigg|_{S_L = \Shigh} = P(\Shigh, \Shigh) + \left[\frac{\partial P}{\partial Q^D}\bigg|_{(\Shigh, \Shigh)} + \frac{\partial P}{\partial S}\bigg|_{(\Shigh, \Shigh)}\right] \cdot \Shigh - G^{-1}(\Shigh)
\]
\[
= \left[\frac{\partial P}{\partial Q^D}\bigg|_{(\Shigh, \Shigh)} + \frac{\partial P}{\partial S}\bigg|_{(\Shigh, \Shigh)}\right] \cdot \Shigh
\]
Since $S_{\text{high}} > 0$, this expression is positive if and only if $\frac{\partial P}{\partial Q^D}\big|_{(S_{\text{high}},S_{\text{high}})} + \frac{\partial P}{\partial S}\big|_{(S_{\text{high}},S_{\text{high}})} > 0$, establishing that the developer's profit is locally increasing at the Regime~II/III boundary under this condition.
 
\paragraph{Proof of part (b).}
 
The large developer compares maximized profit across regimes:
\[
S_L^* = \arg\max\{\pi_L^{I*}, \, \pi_L^{II*}, \, \pi_L^{III*}\}
\]
By part (a), $\pi_L^{II*}$ exceeds profit under every alternative with $S_L < S_{\text{unstable}}$, so $S_L^*$ is either $\Shigh$ (if $\pi_L^{II*} \geq \pi_L^{III*}$) or $S_L^{\text{mon}}$ (if $\pi_L^{III*} > \pi_L^{II*}$). In either case $S_L^* \geq \Shigh$, and the coordination failure is eliminated. \qed
 
\section{Welfare Analysis}
\label{sec:welfare}

The preceding sections characterized the equilibrium behavior of the large and atomistic developers. There is a live question about how efficient these outcomes are. This section derives the social planner's optimum, shows that any atomistic equilibrium underprovisions housing, ranks welfare across equilibria, and compares the Stackelberg outcome to the social optimum.

\subsection{General Welfare Expression}
\label{subsec:welfare_general}

Consider an arbitrary supply level, not necessarily equilibrium, $S$. Both the benevolent social planner and the market will allocate production to the lowest-cost developers. Thus, the marginal entrant has cost $c^* = G^{-1}(S)$. Total welfare at supply level $S$ is given by:
\begin{equation}\label{eq:welfare_general}
    W(S) = \int_0^S P(Q^D, S) \, dQ^D - \int_0^{G^{-1}(S)} c \, g(c) \, dc
\end{equation}

Prices do not appear in the welfare expression because total price paid, $P(S,S) \cdot S$, is subtracted from gross consumer benefit and added to producer surplus, dropping out from the expression. This expression holds for any supply level $S$, inverse demand function $P(Q^D, S)$, and cost distribution $G$.

\subsection{The Social Planner's Problem}
\label{subsec:planner}

A benevolent social planner chooses $S$ to maximize total welfare $W(S)$. Define gross consumer benefit as $B(S) = \int_0^S P(Q^D, S) \, dQ^D$. The supply $S$ appears in $B(S)$ in two places: as the upper limit of integration and as a parameter of the integrand through existing supply's effect on demand. Applying Leibniz's rule:
\begin{equation}\label{eq:dBdS}
    \frac{dB}{dS} = P(S, S) + \int_0^S \frac{\dP}{\dS}\bigg|_{(Q^D, S)} dQ^D
\end{equation}

The first term is the value of the marginal housing unit, the willingness to pay of the consumer for the next unit when total supply is $S$. The second term captures how the addition of one more housing unit affects the willingness to pay of all existing consumers. When $\frac{\dP}{\dS} > 0$, such as when there are positive network effects, this integral is positive.

On the cost side, we differentiate the integral using Leibniz's rule because the upper limit of integration, $c^*$, is a function of $S$ through $c^* = G^{-1}(S)$. Applying Leibniz's rule:
\begin{equation}
    \frac{d}{dS}\int_0^{G^{-1}(S)} c \, g(c) \, dc = G^{-1}(S) \cdot g(G^{-1}(S)) \cdot \frac{1}{g(G^{-1}(S))} = G^{-1}(S)
\end{equation}

This makes intuitive sense: the marginal cost to producers of the extra unit of housing is simply the cost of the next cheapest available developer, which is the developer who builds at the cost threshold $c^* = G^{-1}(S)$.

Putting these two expressions together and setting $\frac{dW}{dS} = 0$, we write the planner's FOC as:
\begin{equation}\label{eq:planner_foc}
    P(S, S) + \int_0^{S} \frac{\dP}{\dS}\bigg|_{(Q^D, S)} dQ^D = G^{-1}(S)
\end{equation}

The planner thus chooses a supply level $S$ such that marginal social benefit (the left side) equals marginal social cost (the right side).

\subsubsection*{Planner's FOC Under Specified Demand}

Under $P(Q^D, S) = \frac{Q_{\max} + \gamma S^\alpha - Q^D}{\beta}$, the network externality integral evaluates to:
\begin{equation}
    \int_0^S \frac{\dP}{\dS} \, dQ^D = \int_0^S \frac{\gamma \alpha S^{\alpha - 1}}{\beta} \, dQ^D = \frac{\gamma \alpha S^\alpha}{\beta}
\end{equation}

Adding to $P(S,S)$, total marginal social benefit is:
\begin{equation}
    \frac{dB}{dS} = \frac{Q_{\max} + \gamma S^\alpha - S}{\beta} + \frac{\gamma \alpha S^\alpha}{\beta} = \frac{Q_{\max} + \gamma(1 + \alpha)S^\alpha - S}{\beta}
\end{equation}

The planner's FOC then becomes:
\begin{equation}\label{eq:planner_foc_specified}
    \frac{Q_{\max} + \gamma(1 + \alpha)S^\alpha - S}{\beta} = G^{-1}(S)
\end{equation}

\subsection{Market Underprovision}
\label{subsec:underprovision}

\begin{lemma}[Equilibrium Does Not Achieve Social Optimum]\label{lem:underprovision}
Suppose that housing supply creates positive network externalities ($\frac{\dP}{\dS} > 0$). Then at any market equilibrium $S^* \in (0, 1)$, welfare is locally increasing: $\frac{dW}{dS}\bigg|_{S = S^*} > 0$. Under the additional regularity condition that $W$ admits a unique interior maximum, this implies $\SFB > S^*$.
\end{lemma}

\begin{figure}[h!]
    \centering
    \includegraphics[width=0.75\textwidth]{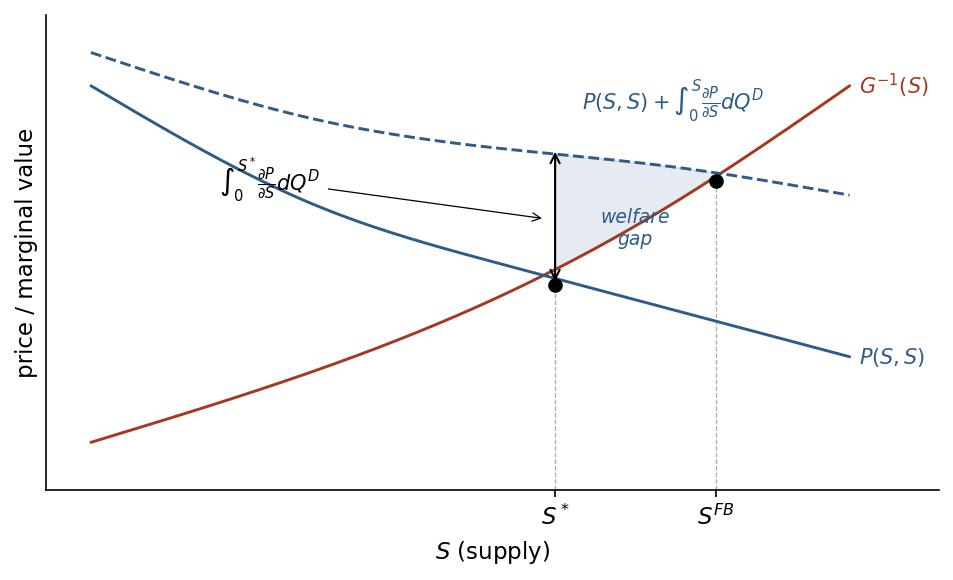}
    \caption{Market underprovision due to network externalities. The market equilibrium $S^*$ equates private marginal benefit $P(S,S)$ with marginal cost $G^{-1}(S)$, while the social planner accounts for the additional benefit $\int_0^S \frac{\partial P}{\partial S} dQ^D$. This creates a wedge between private and social marginal benefit, leading to underprovision relative to the first-best $S^{FB}$.}
    \label{fig:welfare_wedge}
    \vspace{-0.5em}
\end{figure}

\begin{proof}
Evaluate the expression for welfare under a social planner at the market equilibrium $S^*$:
\begin{equation}
    \frac{dW}{dS}\bigg|_{S = S^*} = P(S^*, S^*) + \int_0^{S^*} \frac{\dP}{\dS}\bigg|_{(Q^D, S^*)} dQ^D - G^{-1}(S^*)
\end{equation}

The market equilibrium condition requires that the marginal developer's cost to enter equals their marginal benefit: $P(S^*, S^*) = G^{-1}(S^*)$. Substituting:
\begin{equation}
    \frac{dW}{dS}\bigg|_{S = S^*} = G^{-1}(S^*) + \int_0^{S^*} \frac{\dP}{\dS}\bigg|_{(Q^D, S^*)} dQ^D - G^{-1}(S^*) = \int_0^{S^*} \frac{\dP}{\dS}\bigg|_{(Q^D, S^*)} dQ^D > 0
\end{equation}

This inequality holds because $\frac{\dP}{\dS} > 0$ by assumption. Since $\frac{dW}{dS} > 0$ at equilibrium $S^*$, welfare is strictly increasing in supply at the market equilibrium supply level. Therefore $\SFB > S^*$.
\end{proof}

The implication $\frac{dW}{dS} > 0 \implies S^{FB} > S^*$ requires that $W$ is maximized in $(S^*, 1)$ (the upper bound of supply established in Section~\ref{subsec:primitives}), which holds whenever $W$ is eventually decreasing for large $S$. Under the specified demand function, $\frac{dW}{dS} \to -\infty$ as $S \to 1$ since marginal cost $G^{-1}(S) \to \bar{c}$ while marginal benefit remains bounded, so $W$ is single-peaked and the global maximizer $S^{FB}$ exceeds $S^*$. Under general demand, we maintain as a regularity condition that $W$ admits a unique interior maximum.

\begin{remark}[Welfare Gap Under Specified Demand]\label{rem:welfare_gap}
Under $P(Q^D, S) = \frac{Q_{\max} + \gamma S^\alpha - Q^D}{\beta}$:
\begin{equation}\label{eq:welfare_gap}
    \frac{dW}{dS}\bigg|_{S = S^*} = \frac{\gamma \alpha (S^*)^\alpha}{\beta}
\end{equation}

This is the total uninternalized network externality at the equilibrium supply level $S^*$. Each of the $S^*$ existing consumers experiences $\frac{\gamma \alpha (S^*)^{\alpha - 1}}{\beta}$ increase in welfare with the construction of the next unit of housing. However, none of these existing consumers will pay for the next unit; the marginal developer can only charge one marginal consumer $P(S^*, S^*)$ and has no mechanism to charge for the benefit their unit affords to existing consumers. Thus, the market stops building when private marginal benefit equals private marginal cost, leaving $\frac{\gamma \alpha (S^*)^\alpha}{\beta}$ of welfare on the table. 

This welfare gap is a \citet{spence_1975} distortion. Atomistic entry
decisions internalize only the marginal consumer's willingness to pay
for the next unit, not the inframarginal consumers' network benefits,
and no atomistic entrant has the scale to internalize the latter. The
same logic gives the monopoly platform in \citet{weyl_2010} its residual
welfare loss. In the present setting, the Stackelberg developer partially
closes this gap when it expands into the monopoly region under the
condition of Proposition~\ref{prop:main}(c), since it internalizes
some of its own network externality across its inframarginal units. The
gap closes fully only for a social planner, who values the externality
at all $S$ consumers' willingness to pay.

\end{remark}

\begin{remark}[Comparative Statics of Welfare Gap]\label{rem:welfare_comp_statics}
The welfare that the market equilibrium fails to provide is increasing in $\gamma$ and decreasing in $\beta$. Higher $\gamma$ corresponds to stronger network effects, meaning existing consumers benefit more but the developer has no mechanism to charge for this added welfare. The welfare gap is decreasing in $\beta$ since higher $\beta$ means consumers' demand is more price-elastic; consumers value the added unit of welfare less.
\end{remark}

\subsection{Welfare Ranking Across Equilibria}
\label{subsec:welfare_ranking}

\begin{proposition}[Stackelberg Welfare Improvement] \label{prop:welfare_ranking}
The high equilibrium strictly dominates: $W(\Shigh) > W(\Slow)$. The Stackelberg outcome, in which the large developer eliminates the low equilibrium, strictly improves welfare relative to the low equilibrium outcome.
\end{proposition}

\begin{proof}
Decompose the welfare difference as follows:
\begin{align}
W(\Shigh) - W(\Slow) &= \left[\int_0^{\Shigh} P(Q^D, \Shigh) \, dQ^D - \int_0^{\Slow} P(Q^D, \Slow) \, dQ^D\right] \nonumber \\
&\quad - \left[\int_0^{G^{-1}(\Shigh)} c \, g(c) \, dc - \int_0^{G^{-1}(\Slow)} c \, g(c) \, dc\right] \nonumber \\
&= \underbrace{\int_{\Slow}^{\Shigh} P(Q^D, \Shigh) \, dQ^D}_{\text{(A)}} + \underbrace{\int_0^{\Slow} \left[P(Q^D, \Shigh) - P(Q^D, \Slow)\right] dQ^D}_{\text{(B)}} \nonumber \\
&\quad - \underbrace{\int_{G^{-1}(\Slow)}^{G^{-1}(\Shigh)} c \, g(c) \, dc}_{\text{(C)}}
\end{align}

We show that $\text{(A)} > \text{(C)} > 0$ and $\text{(B)} > 0$.

\textbf{Term (A) vs.\ Term (C).} Since $\frac{\dP}{\dQD} < 0$, the function $P(Q^D, \Shigh)$ is decreasing in $Q^D$. For all $Q^D \in [\Slow, \Shigh]$:
\[
P(Q^D, \Shigh) \geq P(\Shigh, \Shigh) = G^{-1}(\Shigh)
\]
where the equality uses the fact that $\Shigh$ is a market equilibrium. Therefore:
\[
\text{(A)} = \int_{\Slow}^{\Shigh} P(Q^D, \Shigh) \, dQ^D \geq G^{-1}(\Shigh) \cdot (\Shigh - \Slow)
\]
For term (C), every cost $c \in [G^{-1}(\Slow), G^{-1}(\Shigh)]$ satisfies $c \leq G^{-1}(\Shigh)$, with strict inequality since $g > 0$ on the support. Therefore:
\[
\text{(C)} = \int_{G^{-1}(\Slow)}^{G^{-1}(\Shigh)} c \, g(c) \, dc < G^{-1}(\Shigh) \cdot (\Shigh - \Slow)
\]
where we used $\int_{G^{-1}(\Slow)}^{G^{-1}(\Shigh)} g(c) \, dc = \Shigh - \Slow$. Combining: $\text{(A)} > \text{(C)}$.

\textbf{Term (B).} Since $\frac{\dP}{\dS} > 0$ and $\Shigh > \Slow$, we have $P(Q^D, \Shigh) > P(Q^D, \Slow)$ for all $Q^D \in [0, \Slow]$. The integrand is strictly positive, so $\text{(B)} > 0$.

Combining: $W(\Shigh) - W(\Slow) = (\text{A} - \text{C}) + \text{B} > 0$.
\end{proof}

\begin{remark}[Welfare Decomposition]\label{rem:welfare_decomposition}
Welfare at the Stackelberg outcome $S_L^* = \Shigh$ is identical to welfare at the atomistic high equilibrium. The same sites are built at the same costs, and consumers face the same price. The welfare gain from the Stackelberg game relative to the atomistic market is therefore entirely an equilibrium selection effect: the developer guarantees $\Shigh$ rather than risking  coordination failure at $\Slow$. The welfare improvement is exactly  $W(\Shigh) - W(\Slow)$ relative to the low equilibrium and 0 relative to the high equilibrium, so the Stackelberg outcome weakly welfare-dominates every atomistic outcome and strictly dominates whenever $\Slow$ would otherwise have materialized.
\end{remark}

\subsection{Stackelberg vs.\ Social Optimum}
\label{subsec:stackelberg_welfare}

At the Stackelberg equilibrium, total supply is $S^{\text{stack}} = \Shigh$ (assuming the developer does not enter Regime III). Since $\Shigh$ is itself a market equilibrium (it satisfies $P(\Shigh, \Shigh) = G^{-1}(\Shigh)$), Lemma~\ref{lem:underprovision} applies directly:
\begin{equation}
    \frac{dW}{dS}\bigg|_{S = \Shigh} = \int_0^{\Shigh} \frac{\dP}{\dS}\bigg|_{(Q^D, \Shigh)} dQ^D > 0
\end{equation}

Therefore welfare is locally increasing at $\Shigh$, and under the maintained regularity condition of Lemma~\ref{lem:underprovision}, $\SFB > \Shigh$: even the Stackelberg outcome underprovisions housing relative to the social optimum. The Stackelberg leader partially closes the welfare gap (by eliminating the low equilibrium) but does not fully internalize the network externality.

Under specified demand, the remaining welfare gap at the Stackelberg equilibrium is $\frac{\gamma \alpha (\Shigh)^\alpha}{\beta}$.

\begin{remark}[Welfare in the Monopoly Region]\label{rem:monopoly_welfare}
When the large developer enters Regime~III, choosing $S_L^* = S_L^{\text{mon}} > \Shigh$, welfare at the Stackelberg outcome exceeds welfare at the high equilibrium. Since $\frac{dW}{dS} > 0$ at $\Shigh$ (Lemma~\ref{lem:underprovision}), welfare is locally increasing at $\Shigh$. It remains to verify that $S_L^{\text{mon}} < S^{FB}$, so that the developer does not build past the optimum.

Comparing the monopolist's FOC~\eqref{eq:regime3_foc} with the planner's FOC~\eqref{eq:planner_foc}, the monopolist's marginal benefit at any $S$ includes the revenue loss $S \cdot \frac{\partial P}{\partial Q^D} < 0$ on inframarginal units, which the planner omits as a transfer from consumers to the developer. Under the specified demand function,
\[
\frac{\partial P}{\partial S} = \frac{\gamma \alpha S^{\alpha-1}}{\beta}
\]
is constant in $Q^D$, so the monopolist fully internalizes the network externality on its own units:
\[
S \cdot \frac{\partial P}{\partial S}\bigg|_{(S,S)} = \int_0^S \frac{\partial P}{\partial S}\bigg|_{(Q^D, S)}\, dQ^D = \frac{\gamma \alpha S^\alpha}{\beta}.
\]
The wedge between the two FOCs is therefore the standard downward-sloping-demand effect alone. Under general demand, where $\frac{\partial P}{\partial S}$ may vary in $Q^D$, an additional gap arises if the monopolist's network-effect term $S \cdot \frac{\partial P}{\partial S}\big|_{(S,S)}$ falls short of the planner's $\int_0^S \frac{\partial P}{\partial S}\, dQ^D$.

In either case, the monopolist's marginal benefit is strictly below the planner's at any $S$, so $S_L^{\text{mon}} < S^{FB}$. Combined with $S_L^{\text{mon}} > \Shigh$ and the single-peakedness of $W$ maintained in Lemma~\ref{lem:underprovision}, this places $S_L^{\text{mon}}$ on the increasing portion of $W$, so $W(S_L^{\text{mon}}) > W(\Shigh)$. The developer's expansion into the monopoly region partially closes the remaining welfare gap between $\Shigh$ and the social optimum.

\end{remark}

\begin{figure}[h!]
\centering
\begin{subfigure}{0.48\textwidth}
    \centering
    \includegraphics[width=\linewidth]{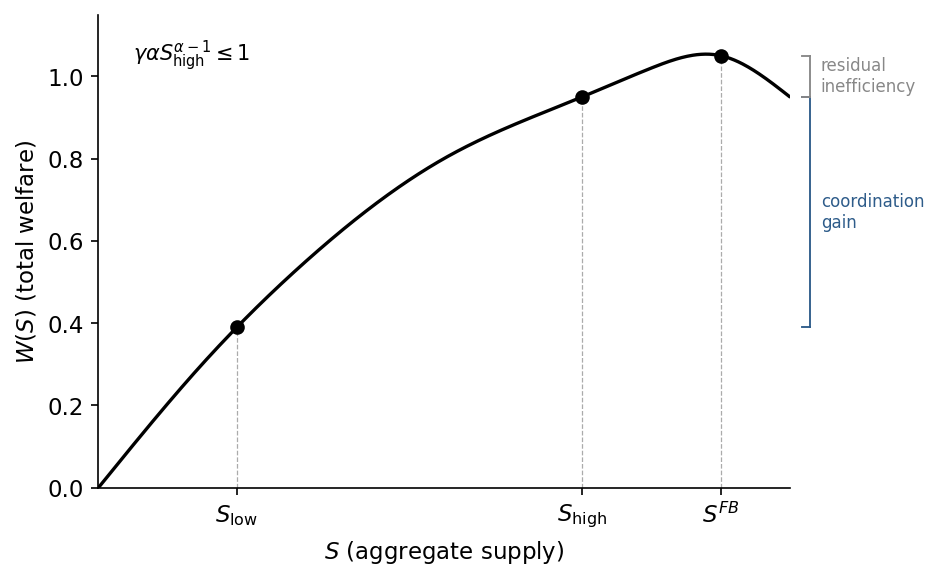}
    \caption{$\gamma \alpha S_{high}^{\alpha - 1} \leq 1$}
\end{subfigure}
\hfill
\begin{subfigure}{0.48\textwidth}
    \centering
    \includegraphics[width=\linewidth]{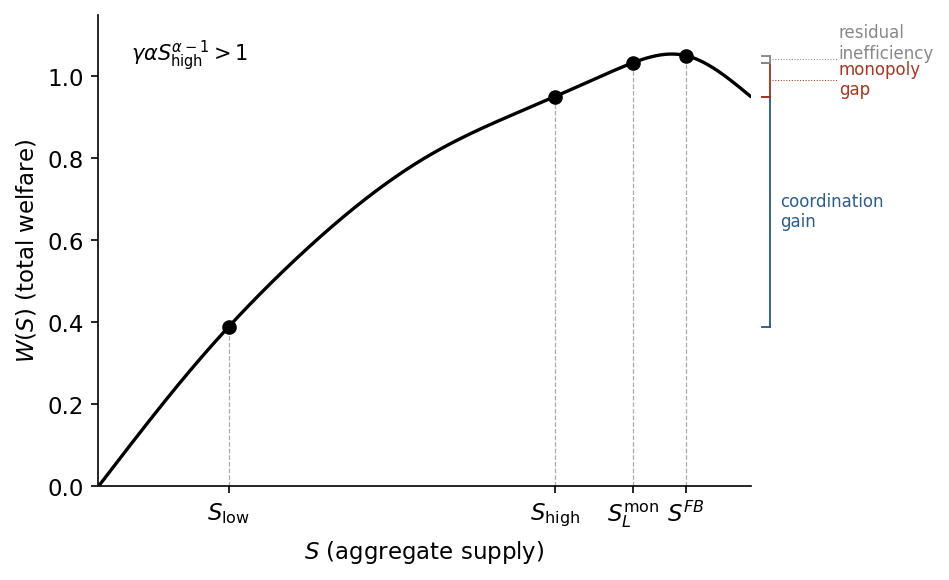}
    \caption{$\gamma \alpha S_{high}^{\alpha - 1} > 1$}
\end{subfigure}

\caption{Welfare effects of Stackelberg leadership under weak and strong network effects. In both cases, eliminating the low equilibrium generates a coordination gain by increasing supply from $S_{low}$ to $S_{high}$. When network effects are weak (left panel), the large developer implements $S_{high}$, leaving a residual inefficiency relative to the social optimum $S^{FB}$. When network effects are strong (right panel), the developer expands into the monopoly region, partially closing the remaining welfare gap.}
\label{fig:welfare_ladder}
\end{figure}

\FloatBarrier
\section{Discussion}
\label{sec:discussion}

\subsection{Robustness to Capacity Constraints}
\label{subsec:displacement}

Within each regime, the large developer optimally claims the entire equilibrium supply, displacing all atomistic developers. This section shows that imposing a capacity constraint $K$ on the number of units the large developer is allowed to claim does not substantially alter the model's results.

Suppose the large developer is constrained to build at most $K$ units, 
where $\Slow < K < \Shigh$. The developer's problem becomes:
\[
\max_{S_L \in [0, K]} \; P(S^*, S^*) \cdot S_L 
- \int_0^{G^{-1}(S_L)} c \, g(c) \, dc
\]
where $S^* \in \mathcal{E}(S_L)$ as before. Since $K > \Slow$, the 
developer can choose $S_L = K > \Slow$, which places the game in 
Regime II. The unique Stage 2 equilibrium is $\Shigh$. Atomistic 
developers build the remaining $\Shigh - K > 0$ units, and total 
supply reaches $\Shigh$. The developer's profit under this choice is:
\[
\pi_L^{K} = G^{-1}(\Shigh) \cdot K 
- \int_0^{G^{-1}(K)} c \, g(c) \, dc
\]

We verify that the developer prefers this to any Regime I outcome. 
Under the $\Slow$ realization, the best Regime I profit is 
$\Pi(\Slow)$. Since $K > \Slow$:
\[
\pi_L^{K} - \Pi(\Slow) = G^{-1}(\Shigh) \cdot K 
- \int_0^{G^{-1}(K)} c \, g(c) \, dc 
- G^{-1}(\Slow) \cdot \Slow 
+ \int_0^{G^{-1}(\Slow)} c \, g(c) \, dc
\]
Since $G^{-1}(\Shigh) > G^{-1}(\Slow)$ and $K > \Slow$, the revenue 
term $G^{-1}(\Shigh) \cdot K$ strictly exceeds 
$G^{-1}(\Slow) \cdot \Slow$. The additional cost 
$\int_{G^{-1}(\Slow)}^{G^{-1}(K)} c \, g(c) \, dc$ is bounded above 
by $G^{-1}(K) \cdot (K - \Slow) < G^{-1}(\Shigh) \cdot (K - \Slow)$, 
since $K < \Shigh$ implies $G^{-1}(K) < G^{-1}(\Shigh)$. Therefore:
\[
\pi_L^{K} > G^{-1}(\Shigh) \cdot \Slow 
- \int_0^{G^{-1}(\Slow)} c \, g(c) \, dc \geq \Pi(\Slow)
\]

Under the $\Shigh$ equilibrium, the best Regime I profit is 
$G^{-1}(\Shigh) \cdot \Slow - \int_0^{G^{-1}(\Slow)} c \, g(c) \, dc$. 
The difference is:
\[
\pi_L^{K} - \pi_L^{I, \text{high}}(\Slow) 
= G^{-1}(\Shigh) \cdot (K - \Slow) 
- \int_{G^{-1}(\Slow)}^{G^{-1}(K)} c \, g(c) \, dc > 0
\]
by the same cost-bounding argument as in Proposition~\ref{prop:main}(a): every 
unit with cost $c \in [G^{-1}(\Slow), G^{-1}(K)]$ satisfies 
$c < G^{-1}(\Shigh)$.

The capacity constraint therefore does not alter the main result. 
For any $K > \Slow$, the developer strictly prefers building $K$ 
units and triggering the high equilibrium. The coordination mechanism 
survives and the developer commits enough to eliminate the low 
equilibrium, and atomistic developers fill in the remaining 
$\Shigh - K$ units. The large developer and atomistic developers 
coexist in equilibrium, which is the empirically relevant case.

\subsection{Alternative Cost Structure: Constant Marginal Cost}
\label{subsec:alt_cost}

The baseline model treats construction costs as the natural consequence of a heterogeneous, observable stock of buildable sites. When the Stackelberg leader optimizes over this stock, they select the lowest-cost sites available. An alternative environment is one in which the leader produces using a separate technology with constant marginal cost $c_L > 0$, while the atomistic fringe continues to draw from the common cost distribution $G$. This subsection characterizes the resulting equilibrium and shows that the coordination result becomes parameter-dependent.

Under this alternative, given a leader commitment $S_L$, atomistic developers enter whenever $c_j \le P(S^*, S^*)$, so the mass of atomistic entrants at price $P(S^*, S^*)$ is $G(P(S^*, S^*))$. Total supply satisfies:

\begin{equation}\label{eq:alt_fixed_point}
    S^* = S_L + G\!\left(P(S^*, S^*)\right).
\end{equation}

Equation \eqref{eq:alt_fixed_point} differs from the baseline fixed-point equation because the leader's supply shifts the fringe response curve directly rather than displacing fringe units one-for-one from a common pool.\footnote{Total supply $S^*$ may exceed 1 under this alternative technology. The normalization in Section~\ref{subsec:primitives} bounded total supply by the mass of sites in the common stock, but here the leader produces outside that stock.} The leader's profit is
\begin{equation}\label{eq:alt_profit}
    \pi_L(S_L, S^*) = \left[P(S^*, S^*) - c_L\right]S_L,
\end{equation}
where $S^*$ is determined by \eqref{eq:alt_fixed_point}. Because $P(S^*, S^*)$ is no longer constant in $S_L$, the leader's objective cannot be analyzed by comparing profits at the fixed baseline prices $G^{-1}(\Slow)$ and $G^{-1}(\Shigh)$. Whether the leader finds it profitable to commit enough supply to move the market into a high-supply region depends jointly on $c_L$, the shape of demand, and fringe response through \eqref{eq:alt_fixed_point}.

The welfare comparison also changes. Under the baseline model, total production cost at supply $S$ is $\int_0^{G^{-1}(S)} c\, g(c)\, dc$ regardless of who builds which sites, because the leader optimally selects the lowest-cost sites and the fringe builds the rest in order. Under constant marginal cost, leader units cost $c_L$ each while fringe developers build according to their own cost draws, giving
\begin{equation}\label{eq:alt_welfare}
    W^{\mathrm{alt}}(S_L, S) = \int_0^S P(Q^D, S)\, dQ^D - c_L \cdot S_L - \int_0^{G^{-1}(S - S_L)} c\, g(c)\, dc.
\end{equation}
Production costs are lower than under the purely atomistic allocation at the same total supply if and only if $c_L \cdot S_L < \int_{G^{-1}(S - S_L)}^{G^{-1}(S)} c\, g(c)\, dc$. The leader must produce more cheaply than the atomistic developers they replace.

Any welfare gain from Stackelberg leadership under this alternative therefore has two distinct sources: a \textit{coordination effect}, from shifting the market toward a higher-supply outcome, and a \textit{production-efficiency effect}, when the leader's constant marginal cost is sufficiently low. The two channels reinforce one another when the leader is efficient and work against one another when the leader is inefficient, so the unconditional result of the baseline does not carry over.

\subsection{Limitations}
\label{subsec:limitations}

The main results rely on three modeling assumptions made for tractability. In each case, the assumption simplifies equilibrium characterization but is not essential to the core intuition of the coordination device. This section discusses what each assumption establishes, what would change if it were relaxed, and why the core result is likely to survive.

\textbf{Simultaneous Stage 2.} We assumed in the model that all atomistic developers move simultaneously in Stage 2. In reality, development occurs sequentially and continuously. After the large developer commits to building, small developers follow over time, and each new entrant's visible commitment raises total supply, shifting the game for later movers.

This assumption simplifies the equilibrium characterization. Simultaneous entry yields the clean fixed-point condition $S^* = G(P(S^*, S^*))$, whereas sequential entry would require modeling each developer's inference about how many predecessors have entered and how many successors will follow. That said, the simultaneous move assumption likely makes the coordination failure harder to resolve than it would be in practice. If atomistic developers entered sequentially, early entrants' visible commitments would signal profitability to later developers, potentially triggering an informational cascade that partially unravels the low equilibrium even without a large developer. The model therefore gives the large developer the most difficult version of the coordination problem, and shows they still solve it. A dynamic multi-period extension is a natural next step. The key open question is whether sequential entry substitutes for the large developer's role or merely complements it. 

\textbf{Common Knowledge of Costs.} The baseline model assumes complete information, that all agents know the realized cost structure and payoff environment. A natural relaxation is to suppose instead that each atomistic developer observes only their own cost draw $c_j$, while the distribution $G$ remains common knowledge. We maintain the baseline assumption that the large developer observes the cost of each potential project. This asymmetry is natural in housing markets. Large developers conduct feasibility studies, commission site-level cost estimates for any potential development, and maintain internal cost databases across their portfolios, all fixed investments in information acquisition that only scale can justify. Atomistic developers, building a single project, rationally know only their own cost. These information-acquisition investments are part of the same scale advantage that justifies the leader's Stage 1 position (Section~\ref{subsec:primitives}).

This private information does not by itself generate aggregate uncertainty. Each atomistic developer, having observed that the Stackelberg leader  selected the $S_L$ lowest-cost sites, knows their own cost exceeds $G^{-1}(S_L)$ and adopts a symmetric cutoff strategy. They build if $c_j \leq P(S^*,S^*)$ and do not enter if not. By the exact law of large numbers, the realized mass of fringe entrants is deterministic and equal to the measure of developers whose costs lie in $(G^{-1}(S_L), P(S^*,S^*)]$. Aggregate entry and the market-clearing price therefore remain deterministic, and the baseline fixed-point condition: 
\[
S^* = G(P(S^*,S^*))
\]
continues to characterize equilibria. The large developer's observable quantity commitment shifts the set of feasible equilibria in exactly the same way as under complete information. 

A genuinely different Bayesian extension would require some source of aggregate uncertainty, such as a finite number of developers, correlated costs, an unknown common fundamental, or uncertainty about the distribution $G$ itself. In such environments, developers would form beliefs over aggregate entry, and Bayesian Nash equilibrium machinery would become necessary.

\textbf{Homogeneous Units.} The last assumption made for tractability was specifying that all units in the market are identical in quality, location, features, amenities, etc and therefore all sell at the same market-clearing price. This assumption ensures we have a single fixed-point condition for all units, as opposed to a vector of prices across housing types and fixed-point conditions for each market segment. The welfare expression would also fragment from a single integral to a sum across segments. While the question of whether constructing housing units of one type (say, market-rate housing) affects (positively or negatively) the price of other types of units (e.g. low-income housing) is hotly debated, the coordination mechanism does not depend on resolving this tension. The core logic -- that density raises willingness to pay through network effects, and that a large developer's commitment can push the market past a low-supply trap -- operates at the neighborhood level regardless of how many unit types exist. Even if a large developer enters a market with heterogeneous housing types but builds exclusively market-rate units, the resulting increase in density raises willingness to pay across all housing types through the network effect, making entry profitable for developers of complementary types and breaking out of the low equilibrium. If anything, heterogeneity may strengthen the coordination failure, since developers specializing in one housing type face additional uncertainty about whether complementary types will enter.

\subsection{Empirical Extension}
\label{subsec:empirical}

The model's predictions can be tested, since its key objects are potentially estimable from housing market data. This section outlines what such an empirical calibration would require and what observable implications it would generate. 

\textbf{Parameter Identification.} The model's key parameters are network effect strength $\gamma$, the convexity parameter $\alpha$, price sensitivity $\beta$, and the cost distribution $G$. As discussed in Section~\ref{sec:literature}, \citet{rossi_hansberg_sarte_owens_2010} estimate the magnitude and spatial decay of housing externalities from parcel-level revitalization programs in Richmond, VA. The Rossi estimation strategy can hence be used to constrain $\gamma$. The price sensitivity parameter $\beta$ can be recovered from hedonic demand estimation using transaction-level data. Regressing observed sale prices on housing characteristics and local supply yields an estimate of how equilibrium prices respond to additional units entering a neighborhood, which corresponds to $\frac{d}{dS}P(S,S)=\frac{\partial P}{\partial Q^D}(S,S)+\frac{\partial P}{\partial S}(S,S)$ in the model. The main identification challenge lies in extricating the demand slope from the network effect, since both affect price when supply changes. The cost distribution $G$ can in principle be estimated from permit-level construction cost records, which are collected by most municipalities. There is an obvious selection issue at play here. The only observed costs are those that correspond to projects that were built, truncating the distribution at the equilibrium entry threshold $P(S^*,S^*)$ since developers do not disclose their costs in the permitting process if they choose not to build. Recovering the full distribution $G$ would therefore require estimating a truncated distribution model that accounts for the known censoring point $P(S^*,S^*)$, or supplementing permit data with developer surveys on projects considered but not pursued. Finally, the convexity parameter $\alpha$ is the hardest to pin down. \citet{couture_handbury_2020} provide qualitative evidence that $\alpha > 1$ might be possible but identifying the precise value of $\alpha$ would require new estimation. Specifically, estimating how the marginal effect of density on willingness to pay varies with the level of density, which is a nonlinear relationship that would need rich micro data on housing transactions across neighborhoods at different density levels. With these parameters estimated, we can test the $\gamma \alpha (S_{\text{high}})^{\alpha - 1} > 1$ threshold condition.

\textbf{Candidate Markets.} The model predicts coordination failures exist in  urban areas with convex network effects; neighborhoods near transit, job centers, or active commercial corridors. A natural experimental design to test the predictions of the model is to compare neighborhoods that received similar regulatory treatment (e.g. upzoning), ideally in the same metro area, but with varying developer structures. If the neighborhood with a large, potentially first-mover developer experienced a discontinuous jump in total supply while the neighborhood without one remained at low supply despite identical regulatory treatment, that is evidence for the coordination mechanism. This design isolates the model's contribution: regulation alone is necessary but not sufficient, and the presence of a large committed developer is what resolves the coordination failure. The Chicago versus Auckland comparison discussed in Section~\ref{sec:intro} is suggestive of this logic, but a more rigorous test would compare neighborhoods within a single metro where $\gamma, \alpha, \beta, \text{and } G$ are plausibly similar.

\textbf{Observable Implications.} Beyond the candidate experimental design above, the model generates several testable predictions that can be evaluated against observed development patterns. First, the effect of developer size on total construction should be nonlinear. Neighborhoods where a developer commits past $S_{\text{unstable}}$ should experience discrete jumps in total supply, rather than the gradual increases implied by a standard competitive model. Second, the relationship between developer market concentration and total housing supply should depend on the local strength of network effects. In dense urban neighborhoods where convex network effects are plausible, 
concentration should be associated with higher total supply, as the 
developer expands beyond $\Shigh$ into the monopoly region. In suburban 
or exurban markets where network effects are weak, the developer 
guarantees $\Shigh$ but does not expand beyond it. Concentration 
changes the composition of who builds (displacing atomistic developers) 
without affecting total supply or prices. Third, neighborhoods that transition from low to high supply should do so rapidly once a sufficiently large developer enters, with atomistic developers following shortly after. Gradual, uniform growth across neighborhoods of similar regulatory environments would be evidence against the coordination mechanism.

\section{Policy Implications}
\label{sec:policy}
While the model specification is rather simplistic and is rarely replicated in complex housing markets, its conclusions yield insights that may be valuable to policymaking bodies. The existence of a market player like the large developer in this paper depends entirely on policy decisions. Local planners can grant certain developers the right of first refusal or other similar incentives to catalyze development in certain areas. If an area is not yet sufficiently desirable for a Stackelberg leader to begin development despite first-mover advantage, local authorities could subsidize construction or backstop losses to incentivize denser construction. In locales with stronger demand for housing, policymakers may simply upzone certain tracts to usher in a flood of development. 

Acknowledging that the conditions for a large developer to play this role may not exist in many municipalities, public-private partnerships could play the first-mover role, committing to initial development to trigger the high equilibrium. A city that constructs public housing as an anchor for a new district or a redevelopment agency that assembles land and installs infrastructure before soliciting private developers. Each of these acts as the Stackelberg leader in the model, absorbing the coordination risk that no individual atomistic developer  would bear alone. The mechanism does not require the public entity to 
build the entire neighborhood; as shown in Section~\ref{subsec:displacement}, it suffices to build past the unstable equilibrium threshold to 
guarantee that private developers fill in the rest.

The model also provides a diagnostic for when concentration generates 
supply beyond what the atomistic high equilibrium would deliver. 
Under specified demand, the large developer builds past the high 
equilibrium when $\gamma \alpha (\Shigh)^{\alpha - 1} > 1$: network effects 
must be strong (high $\gamma$), convex (high $\alpha$), and the market 
must be of sufficient scale (high $S^*$). This condition is most 
plausibly satisfied in dense urban neighborhoods with agglomeration 
benefits. Such neighborhoods include areas near transit, employment 
centers, or existing commercial corridors where each additional unit of 
housing materially improves the neighborhood's amenity value. 
Transit-oriented developments, opportunity zones, and master-planned 
communities are all settings where convex network effects are 
empirically plausible, and where the coordination failure this paper 
identifies is most likely to bind.

When this condition fails, the developer's optimal choice is 
$S_L^* = \Shigh$, and total supply equals exactly what the atomistic 
market would achieve if it coordinated on the high equilibrium. The 
developer's contribution is limited to equilibrium selection: 
guaranteeing $\Shigh$ rather than risking $\Slow$. The sequential 
structure of the game prevents the standard monopoly distortion. If 
the developer chose $S_L < \Shigh$, atomistic developers would fill the 
gap, so supply restriction is not feasible. In suburban or exurban 
markets where density generates few amenity benefits, granting 
first-mover advantage resolves coordination failure but provides no 
additional supply expansion; the welfare gain is entirely the avoided 
loss from the low-supply trap. The key policy implication is therefore 
not that market concentration is universally beneficial, but that its 
benefits are largest in markets where network effects are strong enough 
for the developer to expand supply beyond the high equilibrium, and 
smallest (though still non-negative) in markets where the developer 
merely displaces atomistic builders at $\Shigh$.

\section{Conclusion}
\label{sec:conclusion}
This paper studies coordination failures in housing development markets with network effects. When the marginal benefit of density is convex in supply, a market with exclusively atomistic developers admits multiple stable equilibria, a high-supply equilibrium and a low-supply coordination failure. The specific equilibrium that materializes is completely random, resulting in significant welfare consequences. 

I show that introducing a large developer that moves first and selects sites from the available stock resolves both problems. The large developer always commits at least to the high-supply equilibrium, pushing past the unstable threshold and resolving the coordination failure. This result holds for general demand functions satisfying the regularity conditions of the model and cost distributions, does not depend on any  equilibrium selection assumption, and survives under capacity constraints. Assuming positive network  and negative price effects, the high equilibrium strictly improves welfare compared to the low equilibrium. However, even the large developer underprovisions housing relative to the social optimum because developers cannot charge existing residents for the positive externality (via network effects) that new construction generates. 

The model relies on several tractability assumptions discussed in Section~\ref{subsec:limitations}: simultaneous atomistic entry, complete information, and homogeneous housing units. In each case, the coordination mechanism is likely robust to relaxation, though the formal results would require richer equilibrium concepts. The large developer displaces atomistic developers in equilibrium rather than building alongside them, but the coordination mechanism survives under capacity constraints (Section~\ref{subsec:displacement}). Section~\ref{subsec:alt_cost} considers an alternative environment in which the leader produces with a separate constant-cost technology rather than drawing from the common stock of sites. In that environment, any welfare gain decomposes into a coordination effect and a production-efficiency effect, and the unconditional result no longer carries over.

The results of this paper can be brought to data. Section~\ref{subsec:empirical} outlines a path toward empirical calibration: the network effect parameter $\gamma$ can be bounded by existing externality estimates, the price sensitivity $\beta$ by hedonic demand estimation, and the cost distribution $G$ by permit-level records. The threshold condition $\gamma \alpha (S_{\text{high}})^{\alpha - 1} > 1$ is directly testable once these objects are estimated. A natural design compares similarly regulated neighborhoods with different developer structures to isolate the coordination mechanism from pure regulatory effects.

\bibliographystyle{aer}
\bibliography{references}

@article{rochet_tirole_2003,
    author = {Rochet, Jean-Charles and Tirole, Jean},
    title = {COMPETITION IN TWO-SIDED MARKETS},
    journal = {Journal of the European Economic Association},
    volume = {1},
    number = {4},
    pages = {990--1029},
    year = {2003}
}

@article{armstrong_2006,
    author = {Armstrong, Mark},
    title = {Competition in Two-Sided Markets},
    journal = {The RAND Journal of Economics},
    volume = {37},
    number = {3},
    pages = {668--691},
    year = {2006}
}

@article{rochet_tirole_2006,
    author = {Rochet, Jean-Charles and Tirole, Jean},
    title = {Two-Sided Markets: A Progress Report},
    journal = {The RAND Journal of Economics},
    volume = {37},
    number = {3},
    pages = {645--667},
    year = {2006}
}

@article{weyl_2010,
    author = {Weyl, E. Glen},
    title = {A Price Theory of Multi-Sided Platforms},
    journal = {American Economic Review},
    volume = {100},
    number = {4},
    pages = {1642--1672},
    year = {2010}
}

@article{spence_1975,
    author = {Spence, A. Michael},
    title = {Monopoly, Quality, and Regulation},
    journal = {The Bell Journal of Economics},
    volume = {6},
    number = {2},
    pages = {417--429},
    year = {1975}
}

@article{economides_1996,
    author = {Economides, Nicholas},
    title = {Network Externalities, Complementarities, and Invitations to Enter},
    journal = {European Journal of Political Economy},
    volume = {12},
    number = {2},
    pages = {211--233},
    year = {1996}
}

@article{card_mas_rothstein_2008,
    author = {Card, David and Mas, Alexandre and Rothstein, Jesse},
    title = {Tipping and the Dynamics of Segregation},
    journal = {The Quarterly Journal of Economics},
    volume = {123},
    number = {1},
    pages = {177--218},
    year = {2008}
}

@article{cooper_john_1988,
  author  = {Cooper, Russell and John, Andrew},
  title   = {Coordinating Coordination Failures in {K}eynesian Models},
  journal = {The Quarterly Journal of Economics},
  volume  = {103},
  number  = {3},
  pages   = {441--463},
  year    = {1988}
}

@article{murphy_shleifer_vishny_1989,
  author  = {Murphy, Kevin M. and Shleifer, Andrei and Vishny, Robert W.},
  title   = {Industrialization and the Big Push},
  journal = {Journal of Political Economy},
  volume  = {97},
  number  = {5},
  pages   = {1003--1026},
  year    = {1989}
}

@article{rossi_hansberg_sarte_owens_2010,
  author  = {Rossi-Hansberg, Esteban and Sarte, Pierre-Daniel and Owens, Raymond},
  title   = {Housing Externalities},
  journal = {Journal of Political Economy},
  volume  = {118},
  number  = {3},
  pages   = {485--535},
  year    = {2010}
}

@techreport{owens_rossi_hansberg_sarte_2017,
  author       = {Owens, Raymond and Rossi-Hansberg, Esteban and Sarte, Pierre-Daniel},
  title        = {Rethinking Detroit},
  institution  = {National Bureau of Economic Research},
  type         = {Working Paper},
  number       = {23146},
  year         = {2017},
  url          = {https://www.nber.org/papers/w23146}
}

@article{glaeser_kolko_saiz_2001,
  author  = {Glaeser, Edward L. and Kolko, Jed and Saiz, Albert},
  title   = {Consumer City},
  journal = {Journal of Economic Geography},
  volume  = {1},
  number  = {1},
  pages   = {27--50},
  year    = {2001}
}

@article{couture_handbury_2020,
  author  = {Couture, Victor and Handbury, Jessie},
  title   = {Urban Revival in {A}merica},
  journal = {Journal of Urban Economics},
  volume  = {119},
  pages   = {103267},
  year    = {2020}
}

@book{schelling_1960,
  author    = {Schelling, Thomas C.},
  title     = {The Strategy of Conflict},
  publisher = {Harvard University Press},
  address   = {Cambridge, MA},
  year      = {1960}
}

@article{lipsey_lancaster_1956,
  author  = {Lipsey, R. G. and Lancaster, Kelvin},
  title   = {The General Theory of Second Best},
  journal = {The Review of Economic Studies},
  volume  = {24},
  number  = {1},
  pages   = {11--32},
  year    = {1956}
}

@article{helsley_strange_1994,
  author  = {Helsley, Robert W. and Strange, William C.},
  title   = {City Formation with Commitment},
  journal = {Regional Science and Urban Economics},
  volume  = {24},
  number  = {3},
  pages   = {373--390},
  year    = {1994}
}

@article{guerrieri_hartley_hurst_2013,
  author  = {Guerrieri, Veronica and Hartley, Daniel and Hurst, Erik},
  title   = {Endogenous Gentrification and Housing Price Dynamics},
  journal = {Journal of Public Economics},
  volume  = {100},
  pages   = {45--60},
  year    = {2013}
}

@article{freemark_2020,
  author  = {Freemark, Yonah},
  title   = {Upzoning {C}hicago: Impacts of a Zoning Reform on Property Values and Housing Construction},
  journal = {Urban Affairs Review},
  volume  = {56},
  number  = {3},
  pages   = {758--789},
  year    = {2020}
}

@article{greenaway_mcgrevy_phillips_2023,
  author  = {Greenaway-McGrevy, Ryan and Phillips, Peter C. B.},
  title   = {The Impact of Upzoning on Housing Construction in {A}uckland},
  journal = {Journal of Urban Economics},
  volume  = {136},
  pages   = {103555},
  year    = {2023}
}

\end{document}